\documentclass[11pt,centertags,reqno]{amsart}

\usepackage{latexsym}
\usepackage[english]{babel}
\usepackage[T1]{fontenc}
\usepackage{amssymb}
\usepackage{fancyhdr}
\usepackage{url}
\usepackage{hyperref}
\usepackage{verbatim}
\usepackage{leftidx}

\numberwithin{equation}{section} \makeatletter
\renewcommand{\subsection}{\@startsection
{subsection}{2}{0mm}{\baselineskip}{-0.25cm}
{\normalfont\normalsize\bf}} \makeatother

\newtheorem{theorem}{Theorem}[section]
\newtheorem{lemma}[theorem]{Lemma}

\newtheorem{definition}[theorem]{Definition}
\newtheorem{remark}[theorem]{Remark}
\newtheorem{proposition}[theorem]{Proposition}

\newtheorem{ass}[theorem]{Assumption}

\def \F {\mathcal F}
\def \G {\mathcal G}
\def \H {\mathcal H}

\def \L {\mathcal L}
\def \P {\mathbf P}
\def \Q {\mathbf Q}
\def \R {\mathbb R}

\def \I {{\mathbf 1}}
\def \bF {\mathbb F}
\def \bG {\mathbb G}
\def \bH {\mathbb H}
\def \bE {\mathbb E}
\def \bN {\mathbb N}

\newcommand{\ud}{\mathrm d}
\newcommand{\ds}{\displaystyle}

\newcommand{\espu}[2][\mathbb E^{\P_1}] {#1\left[#2\right]}
\newcommand{\espd}[2][\mathbb E^{\P_2}] {#1\left[#2\right]}
\newcommand{\espp}[2][\mathbb E^{\P^*}] {#1\left[#2\right]}
\newcommand{\espq}[2][\mathbb E^{\Q}] {#1\left[#2\right]}
\newcommand{\espqq}[2][\mathbb E^{\Q^*}] {#1\left[#2\right]}
\newcommand{\condespg}[2][\G_t]{\mathbb E^{\P_2}\left.\left[#2\right|#1\right]}

\newcommand{\condesphh}[2][\widetilde \H_t]       {\mathbb E^{\mathbf Q}\left.\left[#2\right|#1\right]}

\newcommand{\condespgnn}[2][\G^N_t]{\mathbb E^{\P_2} \left.\left[#2\right|#1\right]}

\newcommand{\doleans}[1] {\mathcal E\left(#1\right)}

\author[C.~Ceci]{Claudia  Ceci}
\author[K.~Colaneri]{Katia Colaneri}
\author[A.~Cretarola]{Alessandra Cretarola}

\begin{document}

\address{Claudia  Ceci, Department of Economics,
University ``G. D'Annunzio'' of Chieti-Pescara, Viale Pindaro, 42,
I-65127 Pescara, Italy.}\email{c.ceci@unich.it}

\address{Katia Colaneri, Department of Economics,
University ``G. D'Annunzio'' of Chieti-Pescara, Viale Pindaro, 42,
I-65127 Pescara, Italy.}\email{katia.colaneri@unich.it}

\address{Alessandra Cretarola, Department of Mathematics and Computer Science,
 University of Perugia, Via Vanvitelli, 1, I-06123 Perugia, Italy.}\email{alessandra.cretarola@unipg.it}

\title[Hedging of unit-linked life insurance contracts via LRM]{Hedging of unit-linked life insurance contracts
with unobservable mortality hazard rate via local risk-minimization}

\date{}

\begin{abstract}
\begin{center}
In this paper we investigate the local risk-minimization approach for a combined financial-insurance model where there are restrictions on the information available to the insurance company. In particular we assume that, at any time,  the insurance company may observe the number of deaths from a specific portfolio of insured individuals but not the mortality hazard rate. We consider a financial market driven by a general semimartingale  and we aim to hedge unit-linked life insurance contracts via the local risk-minimization approach under partial information.  The F\"{o}llmer-Schweizer decomposition of the insurance claim and explicit formulas for the optimal strategy  for {\em pure endowment}  and {\em term insurance} contracts are provided in terms of the projection of  the survival process on the information flow. Moreover, in a Markovian framework, we reduce to solve a  filtering problem with point process observations.
\end{center}
\end{abstract}

\maketitle




\section{Introduction}
The paper addresses the problem of computing locally risk-minimizing hedging strategies for unit-linked life insurance contracts under partial information. In these contracts, insurance benefits depend on the price of a specific risky asset and payments are made according to the occurrence of some events related to the stochastic life-length of the policy-holder. In particular we consider a portfolio of $l_a$ insured individuals, having all the same age $a$. Hence, insurance contracts can be considered as contingent claims in an incomplete combined financial-insurance market model, defined on the product of two independent filtered probability spaces: the first one, denoted by $(\Omega_1,\F,\P_1)$ endowed  with a filtration $\bF:=\{\F_t, \ t\geq 0 \}$, is used to model the financial market, while the second one,  $(\Omega_2, \G, \P_2)$ endowed  with a filtration  $ \bG:=\{\G_t, \ t\geq 0 \}$,  describes the insurance portfolio.

In general, incompleteness occurs when the number of assets traded on the market is lower than that of random sources, see e.g. \cite[Chapter 8]{bt}. This is, for instance, the case of insurance claims which are linked to both financial markets and other sources of randomness that are stochastically independent of the financial markets.

Since we consider an insurance market model which is independent of the underlying financial market, we apply the (local) risk-minimization approach for deriving hedging strategies that reduce the risk. This is a quadratic  hedging method which keeps the replication constraint and looks for  hedging strategies (in general not self-financing) with minimal cost.

The concept of risk-minimizing hedging strategies was introduced
in \cite{fs86} in the financial framework only. At the beginning it was formulated assuming that the historical probability measure was a martingale measure. In this context, some results were obtained in the case of full information by  \cite{fs86} and \cite{s01}, and under partial information by \cite{s94} using projection techniques, and more recently in \cite{ccr1} where the authors provide a suitable Galtchouk-Kunita-Watanabe decomposition of the contingent claim that works in a partial information framework. When the historical probability measure does not furnish the martingale measure, i.e. the asset prices dynamics are semimartingales, the theory of risk minimization does not hold and a weaker formulation, namely {\em local risk-minimization}, is required (see e.g. \cite{fs1991,s01}).

Concerning the local risk-minimization approach under partial information, there are less results in the literature, as far as we are aware. In particular, we mention: 
\cite{fs1991},  where the optimal strategy is obtained via predictable projections and
enlargements of filtrations that make the financial market complete; \cite{ccr2}, which provides a suitable F\"{o}llmer-Schweizer decomposition of the contingent claim; \cite{ccc2}, where the authors derive a full description of the optimal strategy under some conditions on the filtrations; an application to the case of defaultable markets in the sense of \cite{fs1991} can be found in \cite{bc}.\\
The (local) risk-minimization approach has also been investigated under the so-called {\em benchmark approach}, a modeling framework that employs the num\'{e}raire portfolio as reference unit, instead of the riskless asset. More precisely, in \cite{bcp} the authors consider the full information setting, the restricted information case is studied  in \cite{CCC}, and finally in \cite{dp} the authors propose a different formulation of the problem for contingent claims which are not square-integrable.

The theory of (local) risk-minimization has been recently extended to the insurance framework, where the market is affected by both mortality and catastrophic risks.
In \cite{moll98} and \cite{vv} the authors study the hedging problem of unit-linked life insurance contracts  in the Black \& Scholes and in the L\'evy financial market model respectively, under full information on both the insurance market and the financial one. In particular, the authors of \cite{vv} discuss the same model analyzed in \cite{ries06}. Moreover, in  \cite{CCC} the problem is solved for a general semimartingale financial model under partial information on the financial market using the benchmark approach. In all these papers lifetimes of the $l_a$ insured individuals are  modeled as i.i.d. non-negative random variables with known hazard rate.

The novelty of this paper consists on considering a combined financial-insurance model where there are restrictions on the information concerning the insurance market. As a matter of fact, we assume that, at time $t$, the insurance company may observe the total number of deaths $N_t$ occurred  until $t$ but not the mortality hazard rate, which depends on an unknown stochastic factor $X$. More precisely, we assume that lifetimes of each individual are conditionally independent, given the whole filtration generated by $X$, $\G^X_{\infty}:=\sigma\{X_u, u \geq 0\}$, and with the same hazard rate process $\lambda_a(t, X_t)$. Denoting by  $\bG^N$ the filtration generated by the process $N$ which counts the number of deaths, then on the combined financial-insurance market 
the information flow available to the insurance company is formally described  by the filtration
$\widetilde \bH := \bF \otimes \bG^N \subseteq  \bH:=\bF\otimes\bG.$

The financial market, on which the insurance company has full knowledge,  consists on a riskless asset with (discounted) price identically equal to $1$ and a risky asset whose (discounted) price $S$ is represented by a semimartingale satisfying the so-called {\em structure condition}, see \eqref{eq:SC}. Since the insurance company's decisions are based on the information flow $\widetilde \bH$, we will look for admissible investment strategies $\psi=(\theta, \eta)$, where the process $\theta$, which describes the amount of wealth invested in the risky asset, is supposed to be $\widetilde \bH$-predictable, whereas the process $\eta$, providing the component invested in the riskless asset, is $\widetilde \bH$-adapted.

We consider two basic forms of insurance contracts, so-called {\em pure endowment} and {\em term insurance}. The policy-holder
of a {\em pure endowment}  contract  receives the payoff  $\xi$ of a contingent claim at a fixed time $T$, if she/he is still alive at this time, while the {\em term insurance} contracts state that the sum insured is due immediately upon death before time $T$. Precisely, payments can occur at any time during $[0,T]$ and are assumed to be time dependent of the form $g(t, S_t)$. In this case the generated obligations are not contingent claims at a fixed time $T$, however they can be transformed into general {\em $T$-claims} by deferring the payments to time $T$.
Under suitable assumptions the payoff of the resulting {\em insurance claim}, denoted by  $G_T$, in both cases is a square-integrable $\widetilde \H_T$-measurable random variable and since the traded asset $S$ turns to be
$\widetilde \bH$-adapted, we can write the F\"{o}llmer-Schweizer decomposition of the random variable $G_T$ with respect to $S$ and $\widetilde \bH$. By applying the results of \cite{s01} and \cite{ccc2} we characterize the pseudo-optimal strategy as the integrand in the  F\"{o}llmer-Schweizer decomposition and the optimal value process as the conditional expected value of the insurance claim $G_T$ with respect to the {\em minimal martingale measure}, given the information flow $\widetilde \H_t$.

In particular, we furnish an explicit formula for the pseudo-optimal strategy in terms of the $\bG^N$-projection of the survival process,
defined in our framework, as ${}_tp_s := \P_2\left(T_i > s+t \ | \ \{T_i>s\} \cap \G^X_\infty\right)$.
In a Markovian setting, its $\bG^N$-optional projection, ${}_t\widehat{p}_s$,  can be written by means of the filter $\pi$, that provides the conditional law of the stochastic factor $X$ given the observed history  $\bG^N$. As a consequence, the computation of the optimal strategy and the optimal value process lead to solve a filtering problem with point process observations.

The literature concerning filtering problems is quite rich, and in particular we can distinguish three main subjects related to different dynamics of the observation process: continuous, counting and mixed type observations.
Counting type observation, which is that considered also in this paper, has been analyzed by \cite{cg97} in the framework of branching processes, and by \cite{cg2000,cg01} for pure jump state processes. An explicit representation of the filter is obtained  in \cite{cg06} by the Feynman-Kac formula using the linearization method introduced by \cite{KMM}. We use this technique to achieve a similar result in our context.
For completeness we indicate some references concerning continuous observation case,  \cite{K, KO, LiS}, and more recently mixed type observation has been studied in \cite{cco1, cco2,fs2012, GrMi}.

The paper is organized as follows. In Section 2 we describe the financial market model. Section 3 is devoted to the insurance market model. In Section 4 we introduce the combined financial-insurance model. The  local risk-minimization is discussed in Section 5.
The F\"{o}llmer-Schweizer decomposition and explicit formulas for the optimal strategy for both {\em pure endowment} and {\em term insurance} contracts are contained in Section 6 and 7, respectively.
Finally, the computation of the survival process and some other technical results are gathered in the Appendix.

\section{The financial market}\label{sec:financial_market}
Let $(\Omega_1,\F,\P_1)$ be a probability space endowed with a filtration $\bF:=\{\F_t, \ t\geq 0 \}$ that satisfies the usual conditions of
right-continuity and completeness; by convention, we set $\F=\F_{\infty}$ and $\F_{\infty-}=\bigvee_{t\geq 0}\F_t$, see e.g. \cite{js}. We consider a simple financial market model where we can find one riskless asset with (discounted) price $1$ and a risky asset whose (discounted) price is represented by an $\R$-valued square-integrable c\`{a}dl\`{a}g $(\bF,\P_1)$-semimartingale $S=\{S_t,\ t \ge 0\}$ that satisfies the so-called {\em structure condition}, i.e.
\begin{equation} \label{eq:SC}
S_t = S_0 + M_t + \int_0^t\alpha_u  \ud \langle M \rangle_u,\quad t \geq 0,
\end{equation}
where $S_0 \in L^2(\F_0,\P_1)$\footnote{The space $L^2(\F_0,\P_1)$ denotes the set of all $\F_0$-measurable random variables $H$ such that $\espu{|H|^2}= \int_\Omega |H|^2\ud {\P_1} < \infty$, where $\mathbb E^{\P^1}[\cdot]$ refers to the expectation computed under the probability measure $\P_1$.}, $M=\{M_t,\ t \geq 0\}$ is an $\R$-valued square-integrable (c\`{a}dl\`{a}g) $(\bF,\P_1)$-martingale
starting at null, $\langle M\rangle=\{\langle M,M\rangle_t,\ t \geq 0\}$ denotes its $\bF$-predictable quadratic variation process and
$\alpha=\{\alpha_t, \ t \geq 0\}$ is an $\R$-valued $\bF$-predictable process such that $\int_0^{T^*}\alpha_s^2 \ud \langle
M\rangle_s < \infty$ $\P_1$-a.s., for each $T^* \in (0,\infty)$.

We assume throughout this paper:
\begin{equation} \label{eq:int_cond}
\espu{\int_0^{T^*}\alpha_u^2 \ud \langle M \rangle_u}<\infty, \quad \forall \ T^*\in (0,\infty).
\end{equation}
Without further mention, all subsequently appearing quantities will be expressed in discounted units.
Since it will play a key role in finding locally risk-minimizing strategies, for reader's convenience we recall the concept of {\em minimal martingale measure} $\P^*$, in short MMM, that is, the unique equivalent martingale measure for $S$ (i.e. $S$ is an $(\bF, \P^*)$-martingale) with the property that $(\bF, \P_1)$-martingales strongly orthogonal to $M$, are also $(\bF, \P^*)$-martingales.

\begin{definition} \label{def:MMM}
Suppose that $S$ satisfies the structure condition. An equivalent martingale measure $\P^*$ for $S$ with square-integrable density $\ds \frac{\ud \P^*}{\ud \P_1}$ is called {\em minimal martingale measure} (for $S$) if $\P^*=\P_1$ on $\F_0$ and if every square-integrable $(\bF, \P_1)$-martingale, strongly orthogonal to the $\P_1$-martingale part of $S$, is also an $(\bF, \P^*)$-martingale.
\end{definition}
If we assume that
  \begin{equation*}\label{eq:hpAS2}
  1-\alpha_t \Delta M_t>0 \quad \P_1-\mbox{a.s.}, \quad \forall \ t \geq 0
  \end{equation*}
and
 \begin{equation}\label{eq:square_integrability_L}
\espu{\exp\left\{\frac{1}{2}\int_0^{T^*} \alpha_t^2 \ud \langle M^c\rangle_t+ \int_0^{T^*} \alpha_t^2 \ud \langle M^d \rangle_t
\right\}} < \infty, \quad \forall \ T^* \in (0,\infty),
  \end{equation}
 where $M^c$ and $M^d$ denote the continuous and the discontinuous parts of the $(\bF,\P_1)$-martingale $M$ respectively and $\alpha$ is given in
 \eqref{eq:SC},
  then by the Ansel-Stricker Theorem, see~\cite{AS}, there exists the MMM $\P^*$ for $S$, which is defined thanks to the density process $L=\{L_t,\ t \geq 0 \}$ given by
\begin{equation}\label{eq:MMMpartial}
L_t:=\left.\frac{\ud \P^*}{ \ud \P_1}\right|_{\F_t}=\doleans{-\int \alpha_u \ud M_u}_t,\quad \forall \ t \geq 0,
\end{equation}
where the notation $\mathcal{E}(Y)$ refers to the Dol\'{e}ans-Dade exponential of an $(\bF,\P_1)$-semimartingale $Y$.

We observe that condition \eqref{eq:square_integrability_L} implies that the nonnegative $(\bF, \P_1)$-local martingale $L$ is indeed a square-integrable
$(\bF, \P_1)$-martingale, see e.g.~\cite{ps2008}.

\section{The insurance market}

The insurance market model is described by a filtered probability space $(\Omega_2, \G, \P_2, \bG)$, where $\bG:=\{\G_t, t\geq 0\}$ is a complete and right-continuous filtration, $\G=\G_\infty$ and $\G_{\infty-}=\bigvee_{t \geq 0} \G_t$. We consider a set of individuals of the same age $a$. Then, we select at random a sample of $l_a$ people. We analyze the case where lifetimes of people in the sample are affected by an unknown stochastic factor, which is represented by a c\`{a}dl\`{a}g stochastic process $X=\{X_t, \ t \geq 0 \}$. Therefore, also inspired by~\cite{fs2012}, we make the following assumption.

\begin{ass} \label{ass:cond_indep_lifetimes}
The remaining lifetimes $T_1,..., T_{l_a}$ of each individual are conditionally independent, doubly-stochastic random times all having the same hazard rate process $\lambda_a:=\{\lambda_a(t,X_t), t \geq 0 \}$. More precisely, we set $\G^X_{\infty}:=\sigma\{X_u, u \geq 0\}$; then there is a measurable function $\lambda_a:[0,\infty)\times \R \to (0,\infty)$, such that
\begin{equation}\label{eq:integrabilita_intensita}
\espd{ \int_0^{T^*}\lambda_a(u,X_u)\ud u} < \infty, \quad \forall \ T^* \in (0,\infty)
\end{equation}
and for all $t_1,t_2,\ldots, t_{l_a} \geq 0$
\begin{equation}\label{eq:intensita}
\P_2\left(T_1>t_1,\ldots, T_{l_a}>t_{l_a}|\G^X_\infty\right)=\prod_{i=1}^{l_a}\exp\left(-\int_0^{t_i}\lambda_a(s, X_s) \ud s\right).
\end{equation}
\end{ass}
Here, $\espd{\cdot}$ denotes the expectation computed under the probability measure $\P_2$.
Note that Assumption \ref{ass:cond_indep_lifetimes} implies that $T_i\neq T_j$ $\P_2-a.s.$ for all $i\neq j$.

We set $R^i_t:=\I_{\{T_i\leq t\}}$, for every $t \geq 0$ and define the filtration $\bG^R:=\{\G^R_t, \ t \geq 0\}$ generated by the vector process $(R^1,\ldots,R^{l_a})$, i.e.
\begin{equation}\label{eq:G^R}
\G^R_t:=\sigma\{R^i_s,\ 0\leq s\leq t,\ i=1,\ldots,l_a\}, \quad \forall \ t \geq 0.
\end{equation}
\begin{remark}\label{rem:Ri}
The process $\widetilde R^i=\{\widetilde R^i_t, \ t\geq 0 \}$, given by $\widetilde R^i_t:=R^i_t-\int_0^{t\wedge T_i} \lambda_a(s, X_s) \ud s$ for every $t \geq 0$, is a $(\bG, \P_2)$-martingale.
\end{remark}
Then, we denote by $N=\{N_t, \ t\geq 0\}$ the process that counts the number of deaths. More precisely, for every $t \geq 0$,
\[
N_t:=\sum_{i=1}^{l_a} \I_{\{T_i\leq t\}}
\]
counts the number of deaths in the time interval $[0,t]$.

Now, let $\bG^N:=\{\G_t^N, t \geq 0 \}$  be the natural filtration of $N$, i.e. $\G_t^N=\sigma\{N_u,\ 0 \leq u \leq t\}$. It is worth stressing that in general $\bG^N\subseteq \bG^R$. Furthermore, we remark that all filtrations are supposed to satisfy the usual conditions.

Note that, as a consequence of Remark \ref{rem:Ri}, $N$ has $(\bG, \P_2)$-predictable intensity $\Lambda=\{\Lambda_t, \ t \geq 0\}$ with
\begin{equation}\label{eq:intensita_N}
\Lambda_t= (l_a- N_{t^-}) \lambda_{a}(t, X_{t^-}), \quad t \geq 0.
 \end{equation}
Hence the process $\widetilde N=\{\widetilde N_t, \ t \geq 0\}$, given by $\widetilde N_t= N_t-\int_0^t\Lambda_r \ \ud r $ for each $t \geq 0$, turns to be a $(\bG,\P_2)$-martingale.
In particular this implies that for every $t \geq 0$,
\[
\espd{\int_0^t\Lambda_r \ud r}=\espd{N_t}\leq l_a.
\]
We introduce the so-called filter $\pi(F)=\{\pi_t(F),\ t \ge 0\}$ defined as
\begin{equation} \label{filter}
\pi_t(F) : =  \espd{F(t,X_t) | \G^N_t }, \quad \forall \ t \geq 0,
\end{equation}
for any measurable function $F(t,x)$ such that $ \espd{ |F(t,X_t) |}< \infty$, for every $t\geq0$.  It is known  that $\pi(F)$ is a probability measure-valued process with c\`{a}dl\`{a}g trajectories (see e.g. \cite{KO}), which provides the conditional distribution of $X$, given the observation $\bG^N$. We denote by $\pi_{t^-}(F)$ its left version.

Then, the $(\bG^N, \P_2)$-predictable intensity of $N$ is $\{\pi_{t^-}(\Lambda), t \geq 0\}$, where
$\pi_{t^-}(\Lambda) = (l_a- N_{t^-})  \pi_{t^-}(\lambda_a)$. Indeed, by Jensen's inequality and \eqref{eq:intensita} we get
\begin{equation}\label{eq:integrab_intensita2}
  \espd{ \int_0^{T^*} \pi_{t^-}(\Lambda) \ud t} \leq \espd{ \int_0^{T^*}\Lambda_t \ud t} \leq l_a  < \infty, \quad \forall \ T^* \in (0,\infty),
  \end{equation}
and for any bounded $(\bG^N, \P_2)$-predictable process $\vartheta=\{\vartheta_t, \ t \geq 0 \}$, we have
\[\espd{ \int_0^{T^*} \vartheta_t  \ud N_t} = \espd{ \int_0^{T^*} \vartheta_t  \Lambda_t  \ud t} = \espd{ \int_0^{T^*} \vartheta_t  \pi_t(\Lambda) \ud t}=
\espd{ \int_0^{T^*} \vartheta_t  \pi_{t^-} (\Lambda) \ud t}, \]
for every $\ T^* \in (0,\infty)$, where the second equality follows by conditioning with respect to the $\sigma$-algebra $\G^N_t$.
Therefore the process $N^*=\{N^*_t,\ t \geq 0\}$, given by $N^*_t:=N_t-\int_0^t\pi_{r^-}(\Lambda)\ud r$, is a $(\bG^N, \P_2)$-martingale.

We define the {\em survival process}
\begin{equation}\label{eq:defn_surv_proc}
{}_tp_s := \P_2\left(T_i > s+t \ | \ \{T_i>s\} \cap \G^X_\infty\right), \quad \forall \ t \geq 0, \ s\geq 0,
\end{equation}
which is a generalization to our context of the well-known {\em survival function} for i.i.d. lifetimes, see e.g. \cite{moll98, moll01, vv} for more details.
It is possible to show that
\begin{equation}\label{eq:surv_proc}
{}_tp_s \I_{\{T_i>s\}} = e^{-\int_s^{s+t}\lambda_a(u, X_u) \ud u}\I_{\{T_i>s\}}, \quad \forall \ t\geq 0, \ s\geq 0.
\end{equation}
See Lemma \ref{lemma:surv} in Appendix \ref{app:survival_process} for the proof.

\section{The combined financial-insurance model}
We now introduce the market model generated by the economy and the insurance portfolio.
In the sequel, we make the following hypothesis.
\begin{ass}\label{ass:indipendenza_assicurazioni}
The insurance market is independent of the a priori given financial market.
\end{ass}
Note that Assumption \ref{ass:indipendenza_assicurazioni} is not restrictive since it covers a wide class of realistic scenarios.
Mathematically, Assumption \ref{ass:indipendenza_assicurazioni} allows to consider the combined financial-insurance model on the following product probability space: $(\Omega, \H, \Q)$, where $\Omega:=\Omega_1\times \Omega_2$, $\H := \F \otimes \G$ and $\Q := \P_1\times\P_2$, endowed with the filtration $\bH:=\bF\otimes\bG$. 

We assume that at any time $t$, the insurance company has a complete information on the financial market and knows the number of policy-holders who are still alive but it cannot observe the hazard rate process $\lambda_a$.  Therefore, the available information flow for the insurance company is formally described by the filtration
\[
\widetilde \bH := \bF \otimes \bG^N.
\]

\begin{remark}\label{rem:variabili_aleatorie}
We observe that for all random variables $Z$ defined on the probability space $(\Omega, \H, \Q)$ there exist two random variables $Z_1$ and $Z_2$ defined  on $(\Omega_1, \F, \P_1)$ and $(\Omega_2, \G, \P_2)$ respectively, such that
$$
Z(\omega_1, \omega_2)=Z_1(\omega_1)Z_2(\omega_2), \quad \forall \ (\omega_1, \omega_2)\in \Omega.
$$
\end{remark}

Finally, by construction, the risky asset price process $S$ turns out to be an $(\bH, \Q)$-semimartingale and an $(\widetilde \bH, \Q)$-semimartingale as well, whose decomposition is given by \eqref{eq:SC} where $S_0 \in L^2(\H_0,\Q)$, $M$ also turns out to be an $\R$-valued square-integrable (c\`{a}dl\`{a}g) $(\bH,\Q)$-martingale and an $(\widetilde \bH, \Q)$-martingale as well.

We put
\begin{equation} \label{def:mmm}
\Q^*:=\P^*\times \P_2
\end{equation}
in the sequel and note that, by construction, $\Q^*$ is an equivalent martingale measure for $S$ with respect to both of the filtrations $\bH$ and $\widetilde \bH$. In particular, it turns out to be the minimal martingale measure for the combined financial-insurance model, as shown in the following lemma.

\begin{lemma}
The minimal martingale measure for $S$ with respect to $\bH$ (and $\widetilde \bH$) exists and coincides with $\Q^*$ given in \eqref{def:mmm}.
\end{lemma}

\begin{proof}
Firstly, we observe that an $(\bH, \Q)$ martingale $O=\{O_t,\ t \ge 0\}$ can be written as $O=O^1 O^2$, where $O^1=\{O^1_t, \ t \geq 0\}$ is an $(\bF, \P_1)$-martingale and
$O^2=\{O^2_t, \ t \geq 0\}$ is a $(\bG, \P_2)$-martingale. Indeed, it is easy to check that every stochastic process living on the the probability space $(\Omega, \H, \Q)$ can be written as the product of two stochastic processes defined on $(\Omega_1, \F, \P_1)$ and $(\Omega_2, \G, \P_2)$ respectively, see also Remark \ref{rem:variabili_aleatorie}. Moreover,
for all $0\leq s<t<\infty$ 
we get
\begin{align*}
\espq{O_t|\H_s}&=
\espq{O^1_t O^2_t|\H_s}
=\espq{O^1_t|\H_s}\espq{O^2_t|\H_s}=\espu{O^1_t |\F_s}\espd{O^2_t |\G_s},
\end{align*}
thanks to Assumption \ref{ass:indipendenza_assicurazioni}.
This implies that $O$ is an $(\bH, \Q)$-martingale if and only if $O^1$ is an $(\bF, \P_1)$-martingale and $O^2$ is a $(\bG, \P_2)$-martingale. Here, $\espq{\cdot|\H_t}$ stands for the conditional expectation with respect to $\H_t$ computed under the probability measure $\Q$ and so on.

Now, let $O$ be a square-integrable $(\bH, \Q)$-martingale strongly orthogonal to $M$. By Assumption \ref{ass:indipendenza_assicurazioni} we have that
\begin{align*}
0 = \langle O,M\rangle_t & = 
\langle O^1O^2, M\rangle_t = O^2_t\langle O^1, M\rangle_t 
\end{align*}
for every $t \geq 0$, and this implies that $O^1$ is strongly orthogonal to $M$.
Since $\P^*$ is the MMM for $S$ with respect to $\bF$, then
$O^1$ is an $(\bF, \P^*)$-martingale, and consequently $O$ is an $(\bH, \Q^*)$-martingale, since $O^2$ is not affected by the change of measure from $\Q$ to $\Q^*$.
This proves that $\Q^*$ is the MMM for $S$ with respect to $\bH$. Finally, note that the same can be done with the filtration $\widetilde \bH$ instead of $\bH$.
\end{proof}

At time $t=0$ the insurance company issues unit-linked life insurance contracts with maturity $T$ for each of the $l_a$ individuals, which are linked to the risky asset price process $S$.

We will consider two basic forms of insurance contracts: the pure endowment and the term insurance. With pure endowment contract, the sum insured is to be paid at time $T$ if the insured is then still alive, while the term insurance states that the sum insured is due immediately upon death before time $T$, see e.g. \cite{moll98} for further details.

The obligation of the insurance company generated by the entire portfolio of pure endowment contracts is given by
\begin{equation}\label{G}
 G_T := \xi \sum_{i=1}^{l_a} \I_{\{T_i>T\}} = \xi (l_a - N_T),
\end{equation}
where $\xi \in L^2(\F_T,\P_1)$ 
represents the payoff of a contingent claim with maturity $T$ on the financial market.
Note that $ G_T\in L^2(\widetilde \H_T, \Q)$.

For term insurance contracts the payment is a time dependent function of the form $g(u,S_u)$, where $g(u,x)$ is a measurable function of its arguments. A simple way of transforming the obligations into a contingent claim with maturity $T$ is to assume that all payment are deferred to the term of the contract $T$. Thus the contingent claim for a portfolio of $l_a$ term insurance contracts is
\begin{equation}\label{G1}
 G_T :=   \sum_{i=1}^{l_a} g(T_i, S_{T_i} )\I_{\{T_i \leq T\}} = \int_0^T g(u,S_u) \ud N_u
\end{equation}
where $g(t,S_t) \in L^2(\F_t,\P_1)$, for every $t \in [0,T]$.
In the sequel we assume that
 \begin{equation}\label{eq:quadrato_integr_sup_g}
 \espu{\sup_{t\in [0,T]}g^2(t, S_t)}<\infty.
  \end{equation}
 This is a classical assumption in the insurance framework (see e.g. \cite{moll01,vv}) and ensures that the random variable $G_T$ is square-integrable, as proved in the following Lemma.
\begin{lemma}
Let the process $\{g(t,S_t),\ t \in [0,T]\}$ satisfy \eqref{eq:quadrato_integr_sup_g}. Then $ G_T\in L^2(\widetilde \H_T, \Q)$.
\end{lemma}

\begin{proof}
\begin{align*}
\espq{G_T^2}&=\espq{\left(\int_0^Tg(t,S_t)\ud N_t\right)^2}=\espq{\left(\sum_{i=1}^{l_a} g(T_i, S_{T_i} )\I_{\{T_i \leq T\}}\right)^2}\\
&\leq l_a \espq{\sum_{i=1}^{l_a} g^2(T_i, S_{T_i} )\I_{\{T_i \leq T\}}}=l_a \espq{\int_0^T g^2(u,S_u) \ud N_u} \\
&= l_a \espq{\int_0^T \left\{g^2(t,S_t)-g^2(t,S_{t^-})\right\} \ud N_t} + l_a \espq{\int_0^Tg^2(t,S_{t^-}) \ud N_t}\\
&=l_a \espq{\int_0^Tg^2(t,S_{t^-}) \ud N_t},
\end{align*}
where the last equality follows by Assumption \ref{ass:indipendenza_assicurazioni},
so that $N$ and $S$ do not have common jump times. Therefore, by \eqref{eq:quadrato_integr_sup_g} we get that
\begin{align*}
l_a \espq{\int_0^Tg^2(t,S_{t^-}) \ud N_t}&=l_a \espq{\int_0^Tg^2(t,S_{t^-}) \Lambda_t \ud t}\leq l_a \espu{\sup_{t \in [0,T]}g^2(t, S_t)}\espd{\int_0^T\Lambda_t \ud t}<\infty.
\end{align*}
\end{proof}
We will refer to $G_T$ as the {\em insurance contingent claim}.

The goal is to construct locally risk-minimizing hedging strategies for the given insurance contingent claims, in the case where the insurance company has  at its disposal the  information flow represented by $\widetilde \bH$.

\section{Locally risk-minimizing strategies under partial information}\label{sec:LRM}

The first step is to introduce the class of all admissible hedging strategies, which have to be adapted to the  information flow $\widetilde \bH$.

\begin{definition}\label{thetaF}
The space $\Theta(\widetilde \bH)$ 
consists of all $\R$-valued $\widetilde \bH$-predictable 
processes $\theta=\{\theta_t,\ t \in [0,T]\}$ satisfying
the following integrability condition:
\begin{equation}\label{admissible}
\espq{\int_0^T\theta_u^2 \ud \langle M\rangle_u+\left(\int_0^T |\theta_u\alpha_u| \ud \langle M \rangle_u\right)^2}<\infty,
\end{equation}
where $\mathbb E^{\Q}[\cdot]$ indicates the expectation with respect to $\Q$.
\end{definition}
Notice that condition \eqref{admissible} implies that the process $\ds \left\{\int_0^t\theta_u \ud S_u, \ t \in [0,T]\right\}$ is well-defined and turns out to be square-integrable.
\begin{definition}
An $\widetilde \bH$-{\em admissible strategy} is a pair $\psi=(\theta,\eta)$, where $\theta \in \Theta(\widetilde \bH)$  and $\eta=\{\eta_t, \ t \in [0,T]\}$ is an
$\R$-valued $\widetilde \bH$-adapted process
such that the value process
$V(\psi)=\{V_t(\psi), \ t \in [0,T]\}:=\theta S + \eta$ is right-continuous and  square-integrable, i.e. $V_t(\psi) \in L^2(\widetilde \H_t,\Q)$, for each $t
\in [0,T]$.
\end{definition}
Here, the processes $\theta$ and $\eta$ represent respectively the units of risky asset and riskless asset held in the portfolio.

For any $\widetilde \bH$-admissible strategy $\psi$, we can define the associated {\em cost process} $C(\psi)=\{C_t(\psi), \ t \in [0,T] \}$, which is the
$\R$-valued $\widetilde \bH$-adapted process given by
$$
C_t(\psi):=V_t(\psi)-\int_0^t \theta_u \ud S_u,
$$
for every $t \in [0,T]$, and the $\widetilde \bH$-{\em risk process} $R^{\widetilde \H}(\psi)=\{R^{\widetilde \H}_t(\psi), \ t \in [0,T]\}$, by setting
\begin{equation*}\label{eq:risk}
R_t^{\widetilde \H}(\psi):=\condesphh{\left(C_T(\psi)-C_t(\psi)\right)^2}
\end{equation*}
for every $t \in [0,T]$. We observe that an $\widetilde \bH$-admissible strategy $\psi$ is self-financing if and only if the associated cost process $C(\psi)$ is
constant and the $\widetilde \bH$-risk process $R^{\widetilde \H}(\psi)$ is zero.

Although $\widetilde \bH$-admissible strategies with $V_T(\psi)=G_T$ will in general not be self-financing, it turns out that good $\widetilde \bH$-admissible strategies are still
self-financing on average in the following sense.
\begin{definition}
An $\widetilde \bH$-admissible strategy $\psi$ is called {\em mean-self-financing} if the associated cost process $C(\psi)$ is an $(\widetilde \bH,\Q)$-martingale.
\end{definition}

The definitions below translate the concept of locally risk-minimizing strategy in the partial information setting.

\begin{definition}
A small perturbation is an $\widetilde \bH$-admissible strategy $\Delta=(\delta^1, \delta^2)$ such that $\delta^1$ is bounded, the variation of $\int\delta^1_u \alpha_u \ud \langle M \rangle_u$ is bounded (uniformly in $t$ and $\omega$) and $\delta^1_T=\delta^2_T=0$. For any subinterval of $[0,T]$, we define the small perturbation
\[
\Delta|_{(s,t]}:=(\delta^1\I_{(s,t]}, \delta^2\I_{[s,t)}).
\]
\end{definition}
We consider a partition $\tau = \{t_i\}_{i=0,1,2,\ldots,k}$ of the time interval $[0,T]$ such that
\[
0=t_0<t_1<\ldots<t_k=T.
\]
\begin{definition}
For an $\widetilde \bH$-admissible strategy $\psi$, a small perturbation $\Delta$ and a partition $\tau$ of $[0,T]$, we set
\[
r^\tau_{\widetilde \H}(\psi, \Delta):=\sum_{t_i,t_{i+1}\in \tau}\frac{R^{\widetilde \H}_{t_i}(\psi+\Delta|_{(t_i,t_{i+1}]})+R^{\widetilde \H}_{t_i}(\psi)}{\bE^\Q\left[\langle M \rangle_{t_{i+1}}-\langle M\rangle_{t_i}|\widetilde{\H}_{t_i}\right]}\I_{(t_i, t_{i+1}]}.
\]
The strategy $\psi$ is called $\widetilde \bH$-locally risk-minimizing if
\[
\liminf_{n\to \infty}r^{\tau_n}_{\widetilde \H}(\psi, \Delta)\geq 0, \quad (\Q\otimes \langle M\rangle)-a.s. \mbox{ on } \Omega\times [0,T],
\]
for every small perturbation $\Delta$ and every increasing sequence $\{\tau_n\}_{n\in \bN}$ of partitions of $[0,T]$ tending to identity.
\end{definition}
Following the idea of \cite{s01}, we now introduce the concept of pseudo-optimal strategy in this framework.
\begin{definition}\label{def:optimalstrategyG}
Let $G_T\in L^2(\widetilde \H_T, \Q)$ be an insurance contingent claim. An $\widetilde \bH$-admissible strategy $\psi$ such that $V_T(\psi)=G_T$ $\Q-\mbox{a.s.}$ is called $\widetilde \bH$-{\em
pseudo-optimal} for $G_T$ if and only if $\psi$ is mean-self-financing and the $(\widetilde \bH,\Q)$-martingale $C(\psi)$ is  strongly orthogonal to the
$\Q$-martingale part $M$ of $S$.
\end{definition}

Note that, $\widetilde \bH$-pseudo-optimal strategies are both easier to find and to characterize, as Proposition \ref{prop:strategy-FS} will show in the following.
In the one-dimensional case and under very general assumptions, locally risk-minimizing strategies and pseudo-optimal strategies coincide, see~\cite[Theorem 3.3]{s01}. More precisely, the equivalence holds if we assume that the predictable quadratic variation $\langle M \rangle$ of the martingale $M$, is $\Q$-a.s. strictly increasing and the finite variation part of $S$ is $\Q$-a.s. continuous. \\

The key result for finding $\widetilde \bH$-pseudo-optimal strategies is represented by the F\"{o}llmer-Schweizer decomposition.
Since the insurance contingent claim $G_T$ belongs to $L^2(\widetilde \H_T, \Q)$, it admits the F\"{o}llmer-Schweizer decomposition with respect to $S$ and $\widetilde \bH$, i.e.
\begin{equation}\label{eq:fsdecompositionG}
G_T = G_0 + \int_0^T \theta^{\widetilde\H}_t \ud S_t + K_T \quad \Q-\mbox{a.s.},
\end{equation}
where $G_0 \in L^2(\widetilde \H_0, \Q)$, $\theta^{\widetilde \H}\in \Theta(\widetilde \bH)$ and $K=\{K_t,\ t \in [0,T]\}$ is a square-integrable $(\widetilde \bH,\Q)$-martingale with $K_0=0$
strongly orthogonal to the $\Q$-martingale part of $S$.

\begin{proposition}\label{prop:strategy-FS}
An insurance contingent claim $G_T \in L^2(\widetilde \H_T, \Q)$ admits a unique  $\widetilde \bH$-pseudo-optimal strategy $\psi^*=(\theta^*,\eta^*)$ with $V_T(\psi^*)=G_T$
$\Q-\mbox{a.s.}$. In terms of the decomposition \eqref{eq:fsdecompositionG}, the strategy $\psi^*$ is explicitly given by
\[
\theta_t^*=\theta^{\widetilde\H}_t, \quad \forall \ t \in [0,T],
\]
with minimal cost
\[
C_t(\psi^*)= G_0+ K_t, \quad \forall \ t \in [0,T];
\]
its value process is
\begin{equation*}\label{eq:optimalport}
V_t(\psi^*)=\condesphh{G_T-\int_t^T \theta^{\widetilde \H}_u \ud S_u}= G_0 + \int_0^t \theta^{\widetilde \H}_u \ud S_u + K_t, \quad \forall \ t \in [0,T],
\end{equation*}
so that $\eta_t^*=V_t(\psi^*)-\theta_t^* S_t$, for every $t \in [0,T]$.
\end{proposition}

\begin{proof}
The proof uses the same arguments of~\cite[Proposition 3.4]{s01}.
\end{proof}

The following result provides an operative way to compute the optimal strategy by switching to the minimal martingale measure $\Q^*$.

\begin{proposition}\label{prop:strategia_ottimaG}
Let $G_T \in L^2(\widetilde \H_T, \Q)$ be an insurance contingent claim and $\psi^*=(\theta^*,\eta^*)$  be the associated $\widetilde \bH$-pseudo-optimal strategy.
Then, the optimal value process $V(\psi^*)=\{V_t(\psi^*),\ t \in [0,T]\}$ is given by
\begin{equation} \label{cla1}
V_t(\psi^*)= \espqq{G_T \Big{|} \widetilde \H_t}, \quad \forall \ t \in [0,T],
\end{equation}
where $\bE^{\Q^*}\left[\cdot \Big{|} \widetilde \H_t\right]$ denotes the conditional expectation with respect to $\widetilde \H_t$ computed under $ \Q^*$; moreover, $\theta^*$ is explicitly given by
\begin{equation} \label{cla2}
\theta_t^*=\frac{ \ud   \langle V(\psi^*), S\rangle_t}{ \ud   \langle S \rangle_t}, \quad \forall \ t \in [0,T] ,
\end{equation}
where the sharp brackets are computed between the martingale parts of the processes $V(\psi^*)$ and $S$, under $\Q$ and with respect to the filtration $\widetilde \bH$.
\end{proposition}
\begin{proof}
The proof follows by that of Proposition 4.2 in \cite{ccc2}.
\end{proof}

\section{Application 1: pure endowment contracts}\label{sec:pure_endowment}
In the sequel we apply the results of Section \ref{sec:LRM} to compute the $\widetilde \bH$-pseudo-optimal strategy $\psi^*$ for a pure endowment contract, whose payoff $G_T$ is given in \eqref{G}. More precisely,  we write the F\"{o}llmer-Schweizer decomposition of the random variable $G_T$ using \eqref{cla1} and apply Proposition \ref{prop:strategy-FS} to identify the optimal strategy.

We observe that the payoff of a contingent claim $\xi \in L^2(\F_T, \P_1)$ admits the F\"{o}llmer-Schweizer decomposition with respect to $S$ and $\bF$, i.e.
\begin{equation}\label{eq:fsdecompositionXi}
\xi = U_0 + \int_0^T \beta_t \ud S_t + A_T \quad \P_1-\mbox{a.s.},
\end{equation}
where $U_0=\espu{\xi|\F_0}\in L^2(\F_0, \P_1)$, $\beta \in \Theta(\bF)$\footnote{The space $\Theta(\bF)$
consists of all $\R$-valued $\bF$-predictable
processes $\delta=\{\delta_t,\ t \in [0,T]\}$ satisfying the integrability condition $\ds \espu{\int_0^T\delta_u^2 \ud \langle M\rangle_u+\left(\int_0^T |\delta_u\alpha_u| \ud\langle M\rangle_u\right)^2}<\infty$.} and $A=\{A_t,\ t \in [0,T]\}$ is a square-integrable $(\bF,\P_1)$-martingale with $A_0=0$  strongly orthogonal to the $\P_1$-martingale part of $S$.

Thanks to Assumption \ref{ass:indipendenza_assicurazioni}, if we take the conditional expectation with respect to $\widetilde \H_t$ under the MMM $\Q^*$ in \eqref{G}, for every $t \in [0,T]$ we get
\begin{equation}\label{eq:valueprocess}
V_t(\psi^*)=\espqq{ G_T \Big{|} \widetilde \H_t} =  \espqq{ \xi \Big{|} \widetilde \H_t} \ \espqq{ l_a - N_T \Big{|} \widetilde \H_t} = \espp{\xi|\F_t} \espd{l_a - N_T|\G^N_t}.
\end{equation}
For the sake of simplicity we define the processes $B=\{B_t,\ t \in [0,T]\}$ and $U=\{U_t,\ t \in [0,T]\}$, given by
\begin{gather}
B_t=:\condespgnn{ l_a - N_T}\label{eq:defn_B},\\
U_t:=\espp{\xi|\F_t},\nonumber
\end{gather}
 for every $t \in [0,T]$.

Note that \eqref{eq:fsdecompositionXi} implies that
\begin{equation}\label{eq:U}
U_t=\espp{\xi|\F_t}=\espp{U_0 + \int_0^T \beta_u \ud S_u + A_T\Big{|}\F_t}=U_0 + \int_0^t \beta_u \ud S_u + A_t,\quad \forall \ t \in [0,T],
\end{equation}
where the last equality holds since $\ds \left\{\int \beta_u \ud S_u, \ t \in [0,T]\right\}$ is an $(\bF, \P^*)$-martingale and $A$ turns out to be an $(\bF, \P^*)$-martingale by definition of MMM.

\begin{proposition}[The $\widetilde \bH$-pseudo-optimal strategy]\label{prop:pseudo_opt_strategy}
The F\"{o}llmer-Schweizer decomposition of the insurance contingent claim $G_T=\xi(l_a-N_T)$ is given by
\[
G_T=G_0+\int_0^TB_{r^-}\beta_r \ud S_r + K_T\quad \Q-\mbox{a.s.},
\]
where $\ds G_0= \espu{\xi | \F_0} \espd{l_a-N_T}$,
\begin{align}\label{eq:mgK}
K_t=\int_0^tB_{s^-} \ud A_s +\int_0^t U_{s^-}\Gamma_{s}\left(\ud N_s-\pi_{s^-}(\Lambda) \ud s\right), \quad \forall \ t \in [0,T],
\end{align}
$B$ is given by \eqref{eq:defn_B}, $\beta$ is the integrand with respect to $S$ in the F\"{o}llmer-Schweizer decomposition of $\xi$, see \eqref{eq:fsdecompositionXi}, and $\Gamma$ is a suitable $(\bG^N, \P_2)$-predictable process, see \eqref{eq:B} below.

Then, the $\widetilde \bH$-pseudo-optimal strategy $\psi^*=(\theta^*, \eta^*)$ is given by
\begin{gather*}
\theta_t^*= B_{t^-} \beta_t, \quad \forall \ t \in [0,T],\\
\eta_t^*=V_t(\psi^*)- B_t \beta_t S_t, \quad  \forall \ t \in [0,T].
\end{gather*}
and the optimal value process $V(\psi^*)$ is given by
\begin{equation}\label{eq:optimalvalueprocess}
V_t(\psi^*)=G_0+\int_0^t B_{r^-}\beta_r \ud S_r + K_t, \quad \forall \ t \in [0,T].
\end{equation}
\end{proposition}

\begin{proof}
Note that, by construction, $B$ is a $(\bG^N, \P_2)$-square-integrable martingale. Then by the Martingale Representation Theorem, see e.g.  \cite[Chapter 4, Theorem 4.37]{js}, there exists a $(\bG^N, \P_2)$-predictable process $\Gamma:=\{\Gamma_t, \ t\in [0,T]\}$ such that $\ds \espd{ \int_0^T\Gamma_u^2 \pi_u(\Lambda)\ud u}<\infty$ and
\begin{align}\label{eq:B}
B_t&=B_0+\int_0^t\Gamma_{u}(\ud N_u-\pi_{u^-}(\Lambda)\ud u)
\end{align}
for every $t \in [0,T]$, where $B_0=\espd{l_a-N_T}=l_a-\espd{\int_0^T\Lambda_s \ud s}$.\\
Therefore, taking \eqref{eq:valueprocess}, \eqref{eq:U} and \eqref{eq:B} into account, by It\^o's product rule we get
\begin{align}
\ud V_t(\psi^*)&= \ud \espqq{G_T \Big{|} \widetilde \H_t}= B_{t^-} \ud U_t + U_{t^-} \ud B_t + \ud \langle B,U\rangle_t + \ud \left(\sum_{s\leq t} \Delta B_s \ \Delta U_s \right)\nonumber\\
& =B_{t^-} \beta_{t} \ud S_t + B_{t^-} \ud A_t + U_{t^-}\Gamma_{t}\left(\ud N_t-\pi_{t^-}(\Lambda) \ud t\right),\label{eq:prodotto}
 \end{align}
since $\langle B,U\rangle_t=0$ for all $t \in [0,T]$ because of the dynamics of $B$, and $\sum_{s\leq t}\left( \Delta B_s \ \Delta U_s \right)=0$ for all $t \in [0,T]$ due to the independence between the financial  and the insurance markets  (see Assumption \ref{ass:indipendenza_assicurazioni}).\\
Clearly, the process $\{B_{t^-}\beta_t,\ t \in [0,T]\}\in \Theta(\widetilde{\bH})$, see Definition \ref{thetaF}. Indeed,
\begin{align*}
\espq{\int_0^TB^2_{t^-}\beta_t^2\ud \langle M \rangle_t + \left(\int_0^TB_{t^-}\beta_t\alpha_t\ud \langle M\rangle_t\right)^2}
\leq l_a^2 \  \espq{\int_0^T \beta_t^2\ud \langle M \rangle_t + \left(\int_0^T\beta_t\alpha_t\ud \langle M\rangle_t\right)^2}<\infty,
\end{align*}
since $0 \leq B_t \leq l_a$, for every $t \in [0,T]$, and $\beta \in \Theta(\widetilde{\bH})$.\\
Finally, since the process $K$ in \eqref{eq:mgK} is a square-integrable $(\widetilde{\bH}, \Q)$-martingale, see Lemma \ref{lemma:mgK} in Appendix \ref{app:aspettazione_condizionata}, we obtain the result.
\end{proof}

In the sequel we provide a characterization of the processes $\beta$ and $B$.

An analogous result to that of Proposition \ref{prop:strategia_ottimaG} can be used to characterize the process $\beta$, which gives the pseudo-optimal strategy (with respect to $\bF$) for the contingent claim $\xi$ in the underlying financial market, see Section \ref{sec:financial_market}. Precisely, the process $\beta$ appearing in the F\"{o}llmer-Schweizer decomposition \eqref{eq:fsdecompositionXi} of $\xi$ can be explicitly characterized thanks to the MMM $\P^*$, as stated below.

\begin{proposition}\label{prop:strategia_ottimaXi}
Let $\xi \in L^2(\F_T, \P_1)$ be a contingent claim and $\phi^*=(\delta^*,\zeta^*)$  be the associated pseudo-optimal strategy (with respect to $\bF$).
Then, the optimal value process $\bar V(\phi^*)=\{\bar V_t(\phi^*),\ t \in [0,T]\}$ is given by
\begin{equation} \label{eq:port_valueXi}
\bar V_t(\phi^*)= \bE^{\P^*}\left[\xi | \F_t\right], \quad \forall \ t \in [0,T],
\end{equation}
where $\bE^{\P^*}\left[\cdot | \F_t\right]$ denotes the conditional expectation with respect to $\F_t$ computed under $ \P^*$. Finally $\beta$ is equal to $\delta^*$, which is explicitly given by
\begin{equation} \label{eq:strategiaBeta}
\beta_t=\delta_t^*=\frac{ \ud   \langle \bar V(\phi^*), S\rangle_t}{ \ud   \langle S \rangle_t}, \quad \forall \ t \in [0,T] ,
\end{equation}
where the sharp brackets are computed under $\P_1$.
\end{proposition}

\begin{proof}
The proof follows by \cite[Proposition 3.4]{s01} and \cite[Proposition 4.2]{ccc2}.
\end{proof}

\begin{remark}
Assume that the contingent claim $\xi$ has the form $\xi=\Upsilon(T, S_T)$. Then:
\begin{enumerate}
  \item in the Black \& Scholes financial market, the hedging strategy is given at any time $t$ by $\beta_t=\frac{\partial \varphi}{\partial s}(t, S_t)$, where $\varphi$ solves the well-known evaluation formula and $\varphi(T, S_T)=\xi$. In particular explicit expressions can be found for Put and Call options;
  \item in the case of incomplete L\'{e}vy driven market models, the strategy $\beta$ has been computed in \cite{vv} (see Theorem 2, equation (34));
  \item computations of the optimal strategy $\beta$ for general semimartingale driven market models can be found in \cite{cvv2010}.
  \end{enumerate}
\end{remark}
The next Lemma gives a representation of the process $B$.

\begin{lemma}
For every $t \in [0,T]$ we define ${}_{T-t}\widehat{p}_t:=\condespgnn{\exp\left\{-\int_{t}^{T}\lambda_a(s,X_s)\ud s\right\}}$.
Then the process $B$ in \eqref{eq:defn_B} is given by
\[B_t:=\condespgnn{ l_a - N_T}=(l_a-N_t){}_{T-t}\widehat{p}_{t}\quad \forall \ t \in [0,T].\]
\end{lemma}

\begin{proof}
We recall that for every $t \in [0,T]$, $B_t=:\condespgnn{ l_a - N_T}$. Then, taking \eqref{eq:G^R} into account, we get
\begin{align}
\condespgnn{ l_a - N_T}&= \condespgnn{\sum_{i=1}^{l_a} \I_{\{T_i>T\}}}=\condespgnn{\sum_{i=1}^{l_a}\espd{\I_{\{T_i>T\}}|\G^{R}_t\vee\G^X_{\infty}}}\nonumber\\
&=\condespgnn{\sum_{i=1}^{l_a}\I_{\{T_i>t\}}\espd{\I_{\{T_i>T\}}\Big{|}\G^{R^i}_t\vee\G^X_{\infty}}}\nonumber\\
&=\condespgnn{\sum_{i=1}^{l_a}\I_{\{T_i>t\}}\espd{\I_{\{T_i>T\}}|\{T_i>t\}\cap\G^X_{\infty}}}\label{eq:pass3}.
\end{align}
The last equality follows by \eqref{eq:asp_cond_T_i} in Appendix \ref{app:aspettazione_condizionata}. By the definition of the survival process \eqref{eq:defn_surv_proc} we get that \eqref{eq:pass3} becomes
\begin{equation}\label{eq:passaggi}
\condespgnn{\sum_{i=1}^{l_a}\I_{\{T_i>t\}}{}_{T-t}p_t}=\condespgnn{\sum_{i=1}^{l_a}\I_{\{T_i>t\}}\exp\left\{-\int_{t}^{T}\lambda_a(s,X_s)\ud s\right\}},
\end{equation}
where the equality follows by \eqref{eq:surv_proc}. Note that the first term in \eqref{eq:passaggi} is also equal to
\begin{align*}
\condespgnn{(l_a-N_t){}_{T-t}p_t} =(l_a-N_t)\condespgnn{{}_{T-t}p_t}=:(l_a-N_t){}_{T-t}\widehat{p}_{t},
\end{align*}
and the second one is equal to
\begin{align*}
\condespgnn{(l_a-N_t)\exp\left\{-\int_{t}^{T}\lambda_a(s,X_s)\ud s\right\}}=(l_a-N_t)\condespgnn{\exp\left\{-\int_{t}^{T}\lambda_a(s,X_s)\ud s\right\}},
\end{align*}
for every $t \in [0,T]$. We refer to ${}_{u}\widehat{p}_{s}$ as the $\bG^N$-projection of the survival process $\{{}_{u}p_s,\ s \ge 0,\ u \ge 0\}$ and we have obtained that for every $t \in [0,T] $, ${}_{T-t}\widehat{p}_t=\condespgnn{\exp\left\{-\int_{t}^{T}\lambda_a(s,X_s)\ud s\right\}}$.

Concluding we have shown that
\[B_t:=\condespgnn{ l_a - N_T}=(l_a-N_t){}_{T-t}\widehat{p}_{t}\quad \forall \ t \in [0,T].\]
\end{proof}
In the sequel we evaluate ${}_{T-t}\widehat{p}_{t}$ in a Markovian setting.

\subsection{The $\bG^N$-projection of the survival process in a Markovian setting}\label{section:GN_conditional_expectation}

On the probability space $(\Omega_2, \G, \P_2)$ we define the process $Y=\{Y_t, t\geq 0\}$, whose dynamics is given by
\begin{equation}\label{eq:Y}
\ud Y_t = -\lambda_a(t, X_t) Y_t \ud t,
\end{equation}
with $Y_0=1$, $\P_2$-a.s., (or equivalently $Y_t=\exp\{-\int_0^{t}\lambda_a(s, X_s) \ud s\}$, for every $t \in [0,T]$). Then, the $\bG^N$-projection of the survival process has the following relationship with the process $Y$:
\[
{}_u\widehat{p}_s=\mathbb E^{\P_2}\left[\left.\frac{Y_{s+u}}{Y_s}\right|\G^N_s\right], \quad u \geq 0, \ s\geq 0,
\]
and in particular ${}_{T-t}\widehat{p}_{t}=\displaystyle \condespgnn{\frac{Y_T}{Y_t}}$, for every $t \in [0, T]$.
\begin{ass}\label{ass:markov}
We assume that the pair $(X,Y)$ is a $(\bG, \P_2)$-Markov process.
\end{ass}

We denote by $\L^{X,Y}$ the Markov generator of the pair $(X,Y)$. Then, by~\cite[Theorem 4.1.7]{ek86} the process
\[
\left\{f(t,X_t,Y_t)-\int_0^t\L^{X,Y}f(u,X_u,Y_u) \ud u,\ t \in [0,T] \right\}
\]
is a $(\bG, \P_2)$-martingale for every function $f(t,x,y)$ in the domain of the operator $\L^{X,Y}$, denoted by $D(\L^{X,Y})$ .

Thanks to the Markovianity of the pair $(X, Y)$, there is a measurable function $h(t,x,y)$ such that
\[
 h(t,X_t, Y_t)=\condespg{\frac{Y_T}{Y_t}}, \quad \forall \ t \in [0, T],
 \]
 and then
 \begin{equation}\label{eq:proiezione_p}
 {}_{T-t}\hat p_t=\espd{h(t,X_t,Y_t)|\G^N_t}, \quad \forall \ t \in [0,T].
 \end{equation}
Since $Y$ is $\bG$-adapted, there is a function $\gamma(t,x,y)$ such that
 \begin{equation}\label{eq:G}
 \gamma(t,X_t,Y_t)=\espd{Y_T|\G_t}, \quad \forall \ t \in [0,T]
 \end{equation}
 and $\ds h(t,X_t, Y_t)=Y_t^{-1} \gamma(t,X_t,Y_t)$, for every $t \in [0,T]$.

Note that the process $\{\gamma(t,X_t,Y_t),\ t \in [0,T]\}$ given in \eqref{eq:G} is a $(\bG, \P_2)$-martingale; 
then we can provide a characterization of the function $\gamma(t,x,y)$ in terms of the solution of a suitable problem with final condition as the following result will show.
\begin{proposition}\label{lemma:caratterizzazioneG}
Let $\widetilde{\gamma}(t,x,y)\in D(\L^{X,Y})$ such that
\begin{equation}\label{eq:problema}
\left\{
\begin{aligned}
\L^{X,Y} \widetilde{\gamma}(t,x,y)&=0\\
\widetilde{\gamma}(T,x,y)&=y.
\end{aligned}
\right.
\end{equation}
Then $\widetilde{\gamma}(t,X_t,Y_t)= \gamma(t,X_t,Y_t)$, for every $t \in [0,T]$, with $\gamma(t,x,y)$ such that \eqref{eq:G} is fulfilled.
\end{proposition}

\begin{proof}
Let $\widetilde{\gamma}(t,x,y)\in D(\L^{X,Y})$ be the solution of \eqref{eq:problema}. By applying It\^{o}'s formula to the process $\widetilde{\gamma}(t,X_t,Y_t)$, we get that the process $\ds \left\{\widetilde{\gamma}(t, X_t, Y_t)- \int_0^t \L^{X,Y}\widetilde{\gamma}(r,X_r,Y_r)\ud r, \ t \in [0,T]\right\}$
is a $(\bG, \P_2)$-martingale. Finally, since $\widetilde{\gamma}(T,X_T,Y_T)=Y_T$, by the martingale property we get the claimed result.
\end{proof}
Now, by \eqref{eq:proiezione_p},  we need to compute the conditional expectation of $h(t,X_t, S_t)$ given the available information $\G^N_t$, for every $t \in[0,T]$. This leads to solve a filtering problem where $N$ is the observation.

For any measurable function $f(t,x,y)$ such that $\espd{|f(t,X_t,Y_t)|}< \infty$, for every $t\in  [0,T]$, we consider the filter $\pi(f)=\left\{\pi_t(f), \ t \geq 0 \right\}$, which is defined by
\begin{equation} \label{def:filtro}
\pi_t(f) : =  \espd{ f(t,X_t,Y_t) | \G^N_t},
\end{equation}
 for each $t \in [0,T]$. We observe that, for any function $f(t,x,y)$ which does not depend on the variable $y$, the filter defined in \eqref{def:filtro} coincides with that given in \eqref{filter}.

 It is known that, in the case of counting observation given by the process $N$,  the filter solves the following Kushner-Stratonovich equation (see e.g. \cite{cco1}):
\begin{equation}\label{eq:KS}
\pi_t(f)=\pi_0(f)+\int_0^t\pi_s(\L^{X,Y} f) \ud s + \int_0^t (\pi_{s^-}(\Lambda))^+ \left\{\pi_{s^-}(\Lambda f)- \pi_{s^-}(\Lambda)\pi_{s^-}(f)+\pi_{s^-}(\bar \L f)\right\}\left(\ud N_s-\pi_{s^-}(\Lambda) \ud s\right),
\end{equation}
for every $f \in D(\L^{X,Y})$ and every $t \geq 0$, where $z^+:= \frac{1}{z} \I_{\{z>0\}}$ and $\bar \L$ is the operator that takes into account possible jump times between $X$ and $N$.

We observe that the $\bG^N$-projection of the survival process can be written in terms of the filter as
\begin{equation}\label{eq:filtro_h}
{}_{T-t}\widehat p_{t}=\pi_t(h) \qquad \forall \ t \in[0, T].
\end{equation}
Then we get the following result which holds in the Markovian case.

\begin{proposition}
Under Assumption \ref{ass:markov}, the $\widetilde{\bH}$-pseudo-optimal strategy $\psi^*=(\theta^*, \eta^*)$ is given by
\begin{align*}
\theta^*_t&=\beta_t(l_a-N_{t^-})\pi_{t^-}(h), \quad \forall \ t \in [0,T]\\
\eta_t^*&= V_t(\psi^*)- \beta_t(l_a-N_t)\pi_{t}(h)S_t, \quad \forall \ t \in [0,T],
\end{align*}
where $\beta$ is given in \eqref{eq:strategiaBeta}, $\pi$ indicates the filter defined in \eqref{def:filtro}, $h(t,x,y)=y^{-1}\gamma(t,x,y)$, $\gamma(t,x,y)$ is the solution of the problem \eqref{eq:problema} and $V(\psi^*)$ is the optimal value process given by \eqref{eq:optimalvalueprocess} in Proposition \ref{prop:pseudo_opt_strategy}, with $B=(l_a-N)\pi(h)$.
\end{proposition}

As mentioned in the Introduction, filtering problems with counting observations have been widely investigated in the literature.  In \cite{cg06}, for example, an explicit representation of the filter is obtained by the Feynman-Kac formula using the linearization method introduced by \cite{KMM}.  In the sequel we apply this procedure to our framework.

To obtain a similar representation for our model we write down the Kushner-Stratonovich equation solved by the filter \eqref{eq:KS},  between two consecutive jump times of the counting process $N$.
We denote by $\{\tau_i\}_{i=1,\ldots,l_a}$ the ordered sequence of the jump times of $N$; then for $t \in [\tau_n, \tau_{n+1})$, equation  \eqref{eq:KS} becomes
\begin{equation}\label{eq:KS1}
\pi_t(f)=\pi_{\tau_n}(f)+\int_{\tau_n} ^t \left( \pi_s(\L_0 f) - \pi_{s}(\Lambda f)- \pi_{s}(\Lambda)\pi_{s}(f) \right) \ud s,
\end{equation}

where we indicate by $\L_0$ the operator given by  $\L_0 f = \L^{X,Y} f - \bar \L f$  for every $f\in D(\L^{X,Y})$, and in particular, if common jumps times between $X$ and $N$ are not allowed, then $\L_0$ coincides with the Markov generator of the pair $(X,Y)$, i.e. $\L_0 = \L^{X,Y}$.

At any jump time of $N$, say $\tau_n$, also the process $\pi$ exhibits a jump,  and its jump-size is given by
\begin{equation}\label{eq:KS3}\pi_{\tau_n}(f) - \pi_{\tau_{n}^-}(f) = (\pi_{\tau_n^-}(\Lambda))^+ \left\{\pi_{\tau_n^-}(\Lambda f)- \pi_{\tau_n^-}(\Lambda)\pi_{\tau_n^-}(f)+\pi_{\tau_n^-}(\bar \L f)\right\}.
\end{equation}

Note that  $\pi_{\tau_n}(f)$ is completely determined by the knowledge of $\pi_{t}(f)$ for every $t$ in the interval  $[\tau_{n-1}, \tau_{n})$. Indeed, for any function $f$,  $\pi_{\tau_{n}^-}(f) = \lim_{ t \to  \tau_{n}^-} \pi_{t}(f).$ Hence, due to the recursive structure of  equation \eqref{eq:KS}, to our aim we only need to solve equation \eqref{eq:KS1} for every $t \in [\tau_n, \tau_{n+1})$ and $n=0,\ldots, l_a-1$.

Now, we denote by $\rho^n=\{\rho^n_t, \ t \in [0,T]\}$ a measure-valued process  that solves the following equation
\begin{equation}\label{eq:KS2}
\rho^n_t (f)=\pi_{\tau_n}(f) +\int_{\tau_n}^t \left[ \rho^n_s(\L_0 f) -\rho^n_s(\Lambda f)\right] \ud s, \quad  t \in  [\tau_n, \tau_{n+1}),  n=0,\ldots, l_a-1
\end{equation}

and such that $\rho^n_t(1) >0$, for every $t \in [0,T]$. It is not difficult to verify that the process $\ds \left\{\frac{\rho^n_t (f)}{\rho^n_t(1)}, \ t  \in [\tau_n, \tau_{n+1})\right\}$ solves equation \eqref{eq:KS1}.
Therefore, if  uniqueness of the solution to the Kushner-Stratonovich equation holds, we can characterize the filter via the linear equation  \eqref{eq:KS2}.
By Theorem 3.3 in \cite{KO} we get the following result.

\begin{proposition}\label{Uni}
Assume that the Martingale Problem for the operator  $\L^{X,Y}$ is well posed and that there exists a domain $D_0$ for $\L^{X,Y}$ such that for every function $ f(t,x,y) \in D_0$,  $\L^{X,Y}f \in C_b([0,T] \times \R \times \R^+)$.
Then weak uniqueness for the Kushner-Stratonovich equation  \eqref{eq:KS} holds.
\end{proposition}

Finally, we assume  that the Martingale Problem for the operator $\L_0$ admits a solution, then, by the Feynman-Kac formula we obtain the following representation for the process $\rho^n(f)$.
Note that the assumption is clearly fulfilled when $X$ and $N$ do not have common jump times, since in that case $\L_0 = \L^{X,Y}$.

\begin{proposition}\label{F-K}
For any $(x,y) \in \R \times \R^+$
assume that the Martingale Problem for the operator $\L_0$ with initial data  $(x,y)$ admits a solution $P_{(x,y)}$ in
 $D_{[0, + \infty)} (\R \times \R^+)$\footnote{$D_{[0, + \infty)} (\R \times \R^+)$ is the space of c\`{a}dl\`{a}g functions from $[0, + \infty)$ into $\R \times \R^+$}.\\
Then a solution of \eqref{eq:KS2}, such that   $\rho^n_t(1) >0$  for any $t \in (\tau_n, \tau_{n+1})$,  is given by
$$
\rho^n_t(f)  = \int_{\R^2} \Psi^n_t(\tau_n, x, y) (f) \pi_{\tau_n} (\ud x,\ud y), \quad \forall \ t \in [0,T],
$$
with $\ds\Psi^n_t(s, x, y) (f)  = \mathbb E^{P_{(x,y)}} \bigg [ f(t, Z^1_t, Z^2_t) \exp\left(- \int_s^t (l_a - n) \lambda_a(r, Z^1_r) \ud r\right) \bigg ]$, where $(Z^1, Z^2)$ denotes the canonical process on  $D_{[0, + \infty)} (\R \times \R^+)$.
\end{proposition}

\begin{proof}
First observe that $\Lambda_t= (l_a - n) \lambda_a(t, X_t)$ for  $t \in  (\tau_n, \tau_{n+1})$. Then the thesis follows by applying Proposition 3.2 in \cite{cg06}.
\end{proof}

\section{Application 2: term insurance contracts}

In this Section we aim to compute the $\widetilde \bH$-pseudo-optimal strategy $\psi^*$ for a term insurance contract, whose payoff, denoted by $G_T$, is given in \eqref{G1}.\\
Similarly to the case of pure endowment contracts, the optimal value process $V(\psi^*)$ is characterized by relationship \eqref{cla1}:
\begin{align*}
V_t(\psi^*) &= \espqq{ G_T \Big{|} \widetilde \H_t} = \espqq{\int_0^T g(u,S_u)\ud N_u \Big{|} \widetilde{H}_t} \nonumber\\
&=\espqq{\int_0^T [ g(u,S_u) - g(u, S_{u^-}) ]  \ud N_u  | \widetilde \H_t}  +
\espqq{ \int_0^T g(u, S_{u^-}) \ud N_u  \Big{|} \widetilde \H_t}\nonumber\\
&=   \int_0^t g(u,S_{u^-}) \ud N_u + \espqq{ \int_t^T g(u,S_{u^-}) \ud N_u  \Big{|} \widetilde \H_t},\nonumber
\end{align*}
for every $t \in [0,T]$, where the last equality is due to the fact that $S$ and $N$ do not have common jump times.  Since the process $\{g(t, S_{t^-}), \ t \in (0,T]\}$ is $\widetilde \bH$-predictable we get
\begin{align}
V_t(\psi^*) &=  \int_0^t g(u,S_{u^-}) \ud N_u + \espqq{ \int_t^Tg(u, S_{u^-})   \Lambda_u \ud u \Big{|} \widetilde \H_t}\nonumber\\
&= \int_0^t g(u,S_{u^-}) \ud N_u  +   \int_t^T  \espqq{ g(u, S_{u})  \Lambda_u  | \widetilde \H_t}  \ud u\nonumber\\
&=  \int_0^t g(u,S_{u^-}) \ud N_u  +   \int_t^T \espp{g(u, S_{u}) |\F_t} \espd{ \Lambda_u | \G^N_t} \ud u\nonumber\\
&= \int_0^t g(u,S_{u^-}) \ud N_u  +   \int_t^T V(u,t) B_t(u) \ud u,\label{eq:prodotto_BuVu}
\end{align}
for every $t \in [0,T]$, where for every $u\in [0,T]$ the processes $B(u)=\{B_t(u), \ t \in [0,u]\}$ and $V(u)= \{V(u,t), \ t \in [0,u]\}$ are respectively given by
\begin{gather}
B_t(u):=\espd{ \Lambda_u | \G^N_t},\quad \forall \ t\in [0,u]  \label{eq:Bu}\\
V(u,t) :=\espp{g(u,S_u) |\F_t} \forall \ t\in [0,u] \label{eq:Vu}.
\end{gather}

We recall that $g(u, S_{u})\in L^2(\F_u, \P_1)$; then, it admits the F\"{o}llmer-Schweizer  decomposition with respect to $S$ which is given by
  \begin{equation}\label{eq:fs_gu}
 g(u,S_u)= V_0 + \int_0^u \beta_r(u) \ud S_r + A^1_u(u)\quad \P_1-\mbox{a.s.}, \forall \ u \in [0,T],
  \end{equation}
where $V_0:=\espu{g(u,S_u)|\F_0}\in L^2(\F_0, \P_1)$, $\beta(u):=\{\beta_t(u), \ t \in [0,u]\}\in \Theta(\bF)$  and $A^1(u):=\{A^1_t(u), \ t \in [0,u]\}$ is a square-integrable $(\bF, \P_1)$-martingale with $A^1_0(u)=0$, strongly orthogonal to $M$.
Since $S$ is an $(\bF, \P^*)$-martingale, thanks to \eqref{eq:fs_gu}, we get that for every $ 0\leq t \leq u\leq T$
\begin{equation}\label{eq:FS_gu_cond}
V(u,t)=\espp{g(u,S_u) |\F_t}=\espp{V_0 + \int_0^u \beta_r(u) \ud S_r + A^1_u(u)\Big|\F_t}=V_0 + \int_0^t \beta_r(u) \ud S_r +  A^1_t(u),
\end{equation}
where  $A^1(u)$ turns out to be an $(\bF, \P^*)$-martingale by definition of MMM.

\begin{ass}
We assume that there exists a process $\gamma(u)=\{\gamma_t(u), 0\leq t\leq u \}$ for every $u \in [0,T]$ such that
\[A^1_t(u)=\int_0^t \gamma_r(u) \ud A^1_r, \quad t \in [0,u],\]
and that $\espu{\int_0^T \gamma^2_t(u) \ud \langle A^1\rangle_t}<\infty$, where $A^1$ is a square-integrable $(\bF, \P_1)$-martingale, strongly orthogonal to $M$.
\end{ass}
This assumption is rather general as it is satisfied whenever a martingale representation theorem for $(\bF, \P_1)$-martingales holds, see e.g. \cite{vv} for L\'{e}vy driven market models.
\begin{proposition}[The $\widetilde \bH$-pseudo-optimal strategy]
The F\"{o}llmer-Schweizer decomposition of the insurance contingent claim $G_T=\int_0^T g(u,S_u)\ud N_u$ is given by
\[
G_T=G_0+\int_0^T  \int_t^T B_{t^-}(u)\beta_t(u) \ud u   \ud S_t + K^1_T \quad \Q-\mbox{a.s.},
\]
where
\begin{align}\label{eq:mgK1}
K^1_t= \int_0^t\left\{g(r, S_{r^-})+\int_r^T  V(u,r^-)\Gamma_r(u)\ud u\right\}\left(\ud N_r-\pi_{r^-}(\Lambda)\ud r\right)+\int_0^t \int_r^T B_{r^-}(u)\gamma_r(u) \ud u \ \ud A^1_r,
\end{align}
for every $t \in [0,T]$, $G_0=\espq{G_T|\widetilde{\H}_0}$, $B(u)$ is given by \eqref{eq:Bu}, $\beta(u)$ is the integrand with respect to $S$ in the F\"{o}llmer-Schweizer decomposition of $g(u, S_u)$, see \eqref{eq:fs_gu}, and $\Gamma(u)$ is a suitable $(\bG^N, \P_2)$-predictable process, see \eqref{eq:rappBu} below.\\
Then, the $\widetilde \bH$-pseudo-optimal strategy $\psi^*=(\theta^*, \eta^*)$ is given by
\begin{gather*}
\theta_t^*= \int_t^T B_{t^-}(u)\beta_t(u) \ud u, \quad \forall \ t \in [0,T],\\
\eta_t^*= V_t(\psi^*)-\left(\int_t^T B_{t^-}(u)\beta_t(u) \ud u\right)S_t , \quad \forall \ t \in [0,T],
\end{gather*}
and the optimal value process $V(\psi^*)$ is given by
\[
V_t(\psi^*)=G_0+\int_0^t\theta^*_r\ud S_r+ K_t^1, \quad \forall \ t \in [0,T].
\]
\end{proposition}

\begin{proof}
Note that, analogously to the previous case of pure endowment contracts, the process $B(u)$  is a square-integrable $(\bG^N, \P_2)$-martingale, and therefore, thanks to the Martingale Representation Theorem, for every $u \in [0,T]$ there exists a $(\G^N, \P_2)$-predictable process $\Gamma(u)=\{\Gamma_t(u), \ t \in [0,u]\}$ such that  $\ds \espd{\int_0^u \Gamma_r(u)^2 \pi_{r}(\Lambda)\ud r}<\infty$ for every $u \in [0,T]$ and
\begin{equation}\label{eq:rappBu}
B_t(u)=B_0(u)+\int_0^t\Gamma_{r}(u) \left(\ud N_r-\pi_{r^-}(\Lambda)\ud r\right), \quad 0\leq t\leq u\leq T,
\end{equation}
where $B_0(u):=\espd{\Lambda_u}$.
We apply the It\^{o} product rule in equation \eqref{eq:prodotto_BuVu} and obtain
\begin{align*}
&\ud V_t(\psi^*)\\
&=g(t,S_{t^-}) \ud N_t - V(t,t) B_{t^-}(t) \ud t + \int_t^T\left\{V(u,t^-)\ud B_{t}(u)+B_{t^-}(u) \ud V(u,t)\right\}\ud u\\
&=g(t,S_{t^-}) \ud N_t - V(t,t) B_{t^-}(t) \ud t +\int_t^T  \left\{ V(u,t^-)\Gamma_t(u)\left(\ud N_t-\pi_{t^-}(\Lambda)\ud t\right)\right\}\ud u\\
&\qquad +\int_t^T\left\{B_{t^-}(u) \beta_t(u) \ud S_t\right\}\ud u  + \int_t^T \left\{B_{t^-}(u)\gamma_t(u)\ud A^1_t \right\}\ud u.
\end{align*}
We recall that $V(t,t)=g(t, S_t)$ and $B_t(t)=\pi_{t}(\Lambda)$; then, more precisely, for every $t \in [0,T]$,  in integral form we get
\begin{align*}
V_t(\psi^*)=&G_0+\int_0^tg(r, S_{r^-})(\ud N_r- \pi_{r^-}(\Lambda)\ud r)+ \int_{t}^T\left\{\int_0^t V(u, r^-)\Gamma_r(u)(\ud N_r- \pi_{r^-}(\Lambda)\ud r)\right\}\ud u\\
&+ \int_t^T\left\{\int_0^tB_{r^-}(u)\beta_r(u)\ud S_r\right\}\ud u+ \int_t^T\left\{\int_0^tB_{r^-}(u)\gamma_r(u)\ud A^1_r\right\}\ud u.
\end{align*}
By Fubini's Theorem we obtain
\begin{align*}
V_t(\psi^*)=&G_0+\int_0^tg(r, S_{r^-})(\ud N_r- \pi_{r^-}(\Lambda)\ud r)+ \int_{0}^t\left\{\int_r^T V(u, r^-)\Gamma_r(u)\ud u\right\}(\ud N_r- \pi_{r^-}(\Lambda)\ud r)\\
&+ \int_0^t\left\{\int_r^TB_{r^-}(u)\beta_r(u)\ud u\right\}\ud S_r+ \int_0^t\left\{\int_r^TB_{r^-}(u)\gamma_r(u)\ud u \right\}\ud A^1_r, \forall \ t \in [0,T] .
\end{align*}
Note that, the process $\ds \left\{\int_t^TB_{t^-}(u) \beta_t(u) \ud u , \ t \in [0,T]\right\}$ is $(\widetilde{\bH}, \Q^*)$-predictable; then to prove that  $\left\{\int_t^TB_{t^-}(u) \beta_t(u) \ud u , \ t \in [0,T]\right\}$ belongs to $\Theta(\widetilde\bH)$, we consider the following estimation
\begin{align*}
&\espq{\int_0^T\left\{\int_r^TB^2_{r^-}(u)\beta^2_r(u)\ud u \right\}\ud \langle M\rangle_r+ \left(\int_0^T\left\{\int_r^TB_{r^-}(u)\beta_r(u)\ud u \right\} \alpha_r \ud \langle M\rangle_r\right)^2}\\
& =\espq{\int_0^T\left\{\int_0^u B^2_{r^-}(u)\beta^2_r(u)\ud \langle M\rangle_r \right\}\ud u + \left(\int_0^T\left\{\int_0^uB_{r^-}(u)\beta_r(u) \alpha_r \ud \langle M\rangle_r \right\} \ud u\right)^2}\\
&\leq \espq{\int_0^T\sup_{0\leq r \leq u}B^2_{r}(u) \left\{\int_0^u \beta^2_r(u)\ud \langle M\rangle_r \right\}\ud u }+ \espq{\left(\int_0^T\sup_{0\leq r\leq u} B_{r}(u) \left\{\int_0^u \beta_r(u) \alpha_r \ud \langle M\rangle_r \right\} \ud u\right)^2}\\
& \leq \int_0^T \espd{\sup_{0\leq r \leq u} B^2_{r}(u)}\espu{\int_0^u \beta^2_r(u)\ud \langle M\rangle_r}\ud u \\
&\qquad + \int_0^T\espd{(\sup_{0\leq r\leq u}|B_{r}(u)|)^2} \espu{(\int_0^u \beta_r(u) \alpha_r \ud \langle M\rangle_r)^2} \ud u<\infty,
\end{align*}
since
\begin{equation}\label{eq:disug_Bu}
 \espd{\sup_{0\leq r\leq u}B^2_{r}(u)}\leq \espd{\left(\sup_{0\leq r\leq u}|B_{r}(u)|\right)^2}\leq c \espd{\langle B(u)\rangle_u}=c\espd{\int_0^u\Gamma^2_r(u)\pi_r(\Lambda)\ud r}<\infty,
\end{equation}
for any positive constant $c$, thanks to the Burkholder-Davis-Gundy inequality, $\beta(u)\in \Theta(\bF)$ and the time interval $[0,T]$ is finite.\\
Finally, since the process $K^1$ in \eqref{eq:mgK1} is a square-integrable $(\widetilde{\bH}, \Q)$-martingale, see Lemma \ref{lemma:mgK1} in Appendix \ref{app:aspettazione_condizionata}, then we get the statement.
\end{proof}
The next result characterizes the process $B(u)$.
\begin{lemma}\label{lemma:Bu}
For every $u \in [0,T]$ the process $B(u)$ defined in \eqref{eq:Bu} is given by
\[
B_t(u):=(l_a-N_t)\espd{\lambda_a(u,X_u)e^{-\int_t^u\lambda_a(r,X_r)\ud r}|\G^N_t}, \quad \forall t \in [0,u].
\]
\end{lemma}

\begin{proof}
Similarly to the pure endowment case, for every $t \in [0,T]$ we get:
\begin{align*}
B_t(u)&=\espd{ ( l_a - N_u) \lambda_a(u,X_u) | \G^N_t}  = \espd{\sum_{i=1}^{l_a}\I_{\{T_i>u\}}\lambda_a(u,X_u)|\G^N_t}\\
&=\espd{\sum_{i=1}^{l_a}\I_{\{T_i>t\}}\espd{\I_{\{T_i>u\}}\Big{|}\G^{R^i}_t\vee\G^X_\infty}\lambda_a(u,X_u)|\G^N_t}\\
&=\espd{\sum_{i=1}^{l_a} \I_{\{T_i>t\}}\lambda_a(u,X_u)e^{-\int_t^u\lambda_a(r,X_r)\ud r}|\G^N_t}\\
&=(l_a-N_t)\espd{\lambda_a(u,X_u)e^{-\int_t^u\lambda_a(r,X_r)\ud r}|\G^N_t}.
\end{align*}
\end{proof}

\subsection{Evaluation of $B(u)$ in a Markovian case}
By Lemma \ref{lemma:Bu} we get that
\begin{align*}
B_t(u)&=(l_a-N_t)\espd{\lambda_a(u,X_u)e^{-\int_t^u\lambda_a(r,X_r)\ud r}|\G^N_t}\\
&=(l_a-N_t)\espd{\lambda_a(u,X_u)\frac{Y_u}{Y_t}\Big{|}\G^N_t}\\
&=(l_a-N_t)\espd{\frac{1}{Y_t}\espd{\lambda_a(u,X_u)Y_u|\G_t}|\G^N_t},
\end{align*}
for every $0\leq t \leq u \leq T$, where $Y$ is the process defined in \eqref{eq:Y}.
We make Assumption \ref{ass:markov}. Then, thanks to the Markovianity of the pair $(X, Y)$, there is a measurable function $H(t,x,y)$ such that
\[
H(t,X_t, Y_t)=\condespg{\lambda_a(u,X_u)\frac{Y_u}{Y_t}}, \quad \forall \ t \in [0, u].
\]
Since $Y$ is $\bG$-adapted, there is a function $k(t,x,y)$ such that
 \begin{equation}\label{eq:k}
k(t,X_t,Y_t)=\espd{\lambda_a(u,X_u)Y_u|\G_t}, \quad \forall \ t \in [0,u]
 \end{equation}
 and $\ds H(t,X_t, Y_t)=Y_t^{-1} k(t,X_t,Y_t)$, for every $t \in [0,u]$.

Since the process $\{k(t,X_t,Y_t),\ t \in [0,T]\}$ given in \eqref{eq:k} is a $(\bG, \P_2)$-martingale, 
analogously to the pure endowment case, we can characterize the function $k(t,x,y)$ in terms of the solution of a suitable problem with final condition as the following result will show.
\begin{proposition}\label{lemma:caratterizzazionek}
Let $\widetilde{k}(t,x,y)\in D(\L^{X,Y})$ such that
\begin{equation}\label{eq:problemak}
\left\{
\begin{aligned}
\L^{X,Y} \widetilde{k}(t,x,y)&=0\\
\widetilde{k}(u,x,y)&=\lambda(u, x )y.
\end{aligned}
\right.
\end{equation}
Then $\widetilde{k}(t,X_t,Y_t)= k(t,X_t,Y_t)$, for every $t \in [0,u]$, with $k(t,x,y)$ such that \eqref{eq:k} is fulfilled.
\end{proposition}

\begin{proof}
The proof follows by that of Lemma \ref{lemma:caratterizzazionek}.
\end{proof}
Note that, due to the final condition of Problem \eqref{eq:problemak}, it is clear that the process $\{k(t, X_t, Y_t), \ t \in [0,u]\}$ depends on $u$.

In terms of the filter, the process $B(u)$ can be written as $B_t(u)=(l_a-N_t)\pi_t(H)$ for every $0\leq t \leq u \leq T$, where the process $\pi(H)$ depends on $u$.

Summarizing, the following result furnishes the optimal hedging strategy for a term structure contract in the Markovian case.

\begin{proposition}
Under Assumption \ref{ass:markov}, the $\widetilde{\bH}$-pseudo-optimal strategy $\psi^*=(\theta^*, \eta^*)$ is given by
\[
\theta^*_t= \int_t^T(l_a-N_{t^-})\pi_{t^-}(H) \beta_t(u) \ud u, \quad \forall \ t \in [0,T],
\]
and $\ds \eta^*_t=V_t(\psi^*)-\left(\int_t^T(l_a-N_{t^-})\pi_{t^-}(H) \beta_t(u) \ud u \right) S_t$,
where $\pi$ indicates the filter defined in \eqref{def:filtro}, $H(t,x,y)=y^{-1}k(t,x,y)$, $k(t,x,y)$ is the solution of problem \eqref{eq:problemak} and $V(\psi^*)$ is the optimal value process given by \eqref{eq:optimalvalueprocess} in Proposition \ref{prop:pseudo_opt_strategy}, with $B(u)=(l_a-N)\pi(H)$.
\end{proposition}


\begin{center}
{\bf Acknowledgements}
\end{center}
The authors have been supported by the Gruppo
Nazionale per l'Analisi Matematica, la Probabilit\`{a} e le loro
Applicazioni (GNAMPA) of the Istituto Nazionale di Alta Matematica
(INdAM).

\bibliographystyle{plain}
\bibliography{RM_biblio}

\appendix
\section{The survival process}\label{app:survival_process}

\medskip

\begin{lemma} \label{lemma:surv}
For every $i=1,\ldots,l_a$, and for every $s\geq 0$, over the event $\{T_i>s\}$, the survival process satisfies
\[
\I_{\{T_i>s\}}(\omega_2){}_tp_s(\omega_2):=\I_{\{T_i>s\}}(\omega_2)\P_2(T_i>s+t|\{T_i>s\}\cap \G^X_\infty)(\omega_2)=e^{-\int_s^{s+t}\lambda_a(u, X_u) \ud u}\I_{\{T_i>s\}}(\omega_2),
\]
for every $t >0$, $\omega_2\in \Omega_2$.
\end{lemma}

\begin{proof}
Note that for every $A\in \G^X_\infty$, $A\cap \{T_i> s\} \in \{T_i>s\} \cap \G^X_\infty$, for every $s\geq 0$. We will show that for every $A\in \G^X_\infty$, every $s\geq 0$ and every $t \geq 0$,
\begin{equation}\label{eq:proba_cond}
\espd{\I_{A\cap\{T_i>s\}} \P_2(T_i > s+t \ | \ \{T_i>s\} \cap \G^X_\infty)}=\espd{\I_{A\cap\{T_i>s\}}e^{-\int_s^{s+t}\lambda_a(u, X_u) \ud u}}.
\end{equation}
Indeed,
\begin{align*}
& \espd{\I_{A\cap\{T_i>s\}} \P_2(T_i > s+t \ | \ \{T_i>s\} \cap \G^X_\infty)}\\
& \quad =\espd{\I_{A\cap\{T_i>s\}}\espd{\I_{\{T_i>s+t\}}|\ \{T_i>s\} \cap \G^X_\infty}}\\
& \quad =\espd{\I_A \I_{\{T_i>s\}}\I_{\{T_i>s+t\}}}= \espd{\I_A\I_{\{T_i>s+t\}}}\\
& \quad =\espd{\I_A \espd{\I_{\{T_i>s+t\}}|\G^X_\infty}}= \espd{\I_A \P_2(T_i>s+t|\G^X_\infty)}\\
& \quad = \espd{\I_A e^{-\int_0^{s+t}\lambda_a(u,X_u)\ud u}}.
\end{align*}
 On the other hand,
\begin{align*}
\espd{\I_{A\cap\{T_i>s\}}e^{-\int_s^{s+t}\lambda_a(u, X_u) \ud u}}&=\espd{\I_{A}\I_{\{T_i>s\}}e^{-\int_s^{s+t}\lambda_a(u, X_u) \ud u}}\\
&=\espd{\I_{A}e^{-\int_s^{s+t}\lambda_a(u, X_u) \ud u} \espd{\I_{\{T_i>s\}}|\G^X_\infty}}\\
&=\espd{\I_{A}e^{-\int_s^{s+t}\lambda_a(u, X_u) \ud u}e^{-\int_0^s\lambda_a(u, X_u) \ud u}}\\
&=\espd{\I_A e^{-\int_0^{s+t}\lambda_a(u,X_u)\ud u}}
\end{align*}
that proves \eqref{eq:proba_cond}. Therefore, we get the result.
\end{proof}

\section{Technical results}\label{app:aspettazione_condizionata}

\medskip

\begin{lemma}
For every $t \in [0,T]$ and for every $i=1,...,l_a$ the following equality holds
\begin{equation}\label{eq:asp_cond_T_i}
\I_{\{T_i>t\}}\espd{\I_{\{T_i>T\}}|\G^{R^i}_t\vee\G^X_{\infty}}=\I_{\{T_i>t\}}\espd{\I_{\{T_i>T\}}|\{T_i>t\}\cap\G^X_{\infty}}.
\end{equation}
\end{lemma}

\begin{proof}
First we observe that $\G^{R^i}_t\vee\G^X_{\infty}\supseteq \{T_i>t\}\cap\G^X_{\infty}$. Then to prove the equality it is sufficient to show that
\[
\espd{\I_{B}\I_{\{T_i>t\}}\espd{\I_{\{T_i>T\}}|\G^{R^i}_t\vee\G^X_{\infty}}}=\espd{\I_B\I_{\{T_i>t\}}\espd{\I_{\{T_i>T\}}|\{T_i>t\}\cap\G^X_{\infty}}}
\]
for every $t \in [0,T]$ and $B\in \G^{R^i}_t\vee\G^X_{\infty}$. In particular we can only show that the equality holds for the events $B_1=A\cap \{T_i>s\}$ and $B_2=A\cap\{T_i\leq s\}$, for every $s\leq t$, $A\in \G^X_{\infty}$. We get:
\begin{align*}
\espd{\I_{B_1}\I_{\{T_i>t\}}\espd{\I_{\{T_i>T\}}|\G^{B^i}_t\vee\G^X_{\infty}}}&=\espd{\I_{A}\I_{\{T_i>s\}}\I_{\{T_i>t\}}\espd{\I_{\{T_i>T\}}|\G^{R^i}_t\vee\G^X_{\infty}}}\\
&=\espd{\espd{\I_{A}\I_{\{T_i>s\}}\I_{\{T_i>t\}}\I_{\{T_i>T\}}|\G^{R^i}_t\vee\G^X_{\infty}}}\\
&=\espd{\I_A\I_{\{T_i>T\}}}.
\end{align*}
On the other hand,
\begin{align*}
\espd{\I_{B_1}\I_{\{T_i>t\}}\espd{\I_{\{T_i>T\}}|\{T_i>t\}\cap\G^X_{\infty}}}&=\espd{\I_A\I_{\{{T_i>s}\}}\I_{\{T_i>t\}}\espd{\I_{\{T_i>T\}}|\{T_i>t\}\cap\G^X_{\infty}}}\\
&=\espd{\I_A\I_{\{T_i>t\}}\espd{\I_{\{T_i>T\}}|\{T_i>t\}\cap\G^X_{\infty}}}\\
&=\espd{\espd{\I_A\I_{\{T_i>t\}}\I_{\{T_i>T\}}|\{T_i>t\}\cap\G^X_{\infty}}}\\
&=\espd{\I_A\I_{\{T_i>T\}}},
\end{align*}
which proves the statement for the event $B_1$.
Finally for the event $B_2=A\cap\{T_i\leq s\}$, since $s< t$, we get $\I_{\{T_1\leq s\}}\I_{\{T_i>t\}}=0$, and this concludes the proof.
\end{proof}

\medskip

\begin{lemma}\label{lemma:mgK}
Let $K:=\{K_t, \ t \in [0,T]\}$ be the process defined in \eqref{eq:mgK}, i.e.
\begin{align*}
K_t=\int_0^tB_{s^-} \ud A_s +\int_0^t U_{s^-}\Gamma_{s}\left(\ud N_s-\pi_{s^-}(\Lambda) \ud s\right), \quad \forall \ t \in [0,T],
\end{align*}
where $B$ is given by \eqref{eq:defn_B} and $\Gamma$ is defined in \eqref{eq:B}.
Then $K$ is a square-integrable $(\widetilde{\bH}, \Q)$-martingale.
\end{lemma}

\begin{proof}

By the boundedness of $B$, we get that the process $\ds \left\{\int_0^t B_{s^-} \ud A_s, \ t \in [0,T]\right\}$ turns out to be an $(\widetilde \bH, \Q)$-martingale strongly orthogonal to $S$. Moreover, we also get that the process  $\ds \left\{\int_0^t U_{s^-}\Gamma_{s}\left(\ud N_s- \pi_{s^-}(\Lambda) \ud s\right), t \in [0,T]\right\}$ is  an $(\widetilde \bH, \Q)$-martingale strongly orthogonal to $S$. Indeed,
\begin{align}\label{eq:stime1}
\espq{ \int_0^T |U_{s^-}\Gamma_{s}| \pi_{s^-}(\Lambda) \ud s }\leq\espu{\sup_{0\leq t\leq T}|U_{t}|}\espd{ \int_0^T |\Gamma_{s}| \pi_{s^-}(\Lambda) \ud s }.
\end{align}
By \eqref{eq:U} and the structure condition \eqref{eq:SC},  $U_t=U_0+\int_0^t\beta_u \ud M_u+\int_0^t\beta_u \alpha_u \ud \langle M \rangle_u+ A_t$, for every $t \in [0,T]$, which means that
\begin{align}
\sup_{0\leq t \leq T}|U_t|\leq |U_0|+\sup_{0\leq t \leq T}|\int_0^t\beta_u \ud M_u|+\sup_{0\leq t \leq T}\int_0^t|\beta_u \alpha_u| \ud \langle M \rangle_u+ \sup_{0\leq t \leq T}|A_t|, \quad \forall \ t \in [0,T].
\end{align}
Therefore thanks to the Burkholder-Davis-Gundy inequality there exists a constant $c>0$ such that
\begin{align*}
\espq{\sup_{0\leq t\leq T}|U_{t}|}=&\espu{\sup_{0\leq t\leq T}|U_{t}|}\leq |U_0|+ c\espu{\left(\int_0^T\beta_u^2\ud \langle M\rangle_u\right)^{\frac{1}{2}}}\\
&+\frac{1}{2}\espu{\int_0^T\beta_u^2\ud \langle M\rangle_u}+\frac{1}{2}\espu{\int_0^T\alpha_u^2\ud \langle M\rangle_u}+c\espu{\left(\langle A\rangle_T\right)^{\frac{1}{2}}}<\infty
\end{align*}
Then the expectations in \eqref{eq:stime1} are finite.

Finally, we observe that $K$ is square-integrable. In fact, since $K$ is an $(\widetilde{\bH}, \Q)$-martingale, then $K^2$ is an $(\widetilde{\bH}, \Q)$-submartingale and for every $t \in [0,T]$
\[
\espq{K_t^2}\leq \espq{K_T^2}\leq3\left(\espq{G_T^2}+ \espq{\left(\int_0^TB_{t^-}\beta_t \ud S_t\right)^2}+\espq{G_0^2}\right)<\infty
\]
since $G_T\in L^2(\widetilde{\H}_T, \Q)$, $G_0\in L^2(\widetilde{\H}_0, \Q)$ and $\{B_{t^-}\beta_t, \ t \in [0,T]\}\in \Theta(\widetilde\bH)$.
\end{proof}

\medskip

\begin{lemma}\label{lemma:mgK1}
Let $K^1:=\{K^1_t, t \in [0,T]\}$ be the process defined in \eqref{eq:mgK1}, i.e.
\begin{align*}
K^1_t= \int_0^t\left\{g(r, S_{r^-})+\int_r^T  V(u,r^-)\Gamma_r(u)\ud u\right\}\left(\ud N_r-\pi_{r^-}(\Lambda)\ud r\right)+\int_0^t \int_r^T B_{r^-}(u)\gamma_r(u) \ud u \ \ud A^1_r,
\end{align*}
for every $ t \in [0,T]$, where $B(u)$ is given by \eqref{eq:Bu}  and $\Gamma(u)$ is defined in \eqref{eq:rappBu}.
Then $K^1$ is a square-integrable $(\widetilde{\bH}, \Q)$-martingale.
\end{lemma}

\begin{proof}
Firstly we note that the process $\ds\left\{\int_0^t g(r, S_{r^-}) \left(\ud N_r - \pi_{r^-}(\Lambda) \ud r\right), \ t \in [0,T]\right\}$ is a square-integrable $(\widetilde{\bH}, \Q)$-martingale. Indeed
\[
\espq{\int_0^T g^2(r, S_r) \pi_r(\Lambda) \ud r}\leq\espu{\sup_{t\in [0,T]}g^2(t, S_t)}\espd{\int_0^T \Lambda_r \ud r}\leq l_a \espu{\sup_{t\in [0,T]}g^2(t, S_t)}<\infty.
\]

Second, since $\{V(u,t^-), \ t \in [0,u]\}$ is $\bF$-predictable and $\Gamma(u):=\{\Gamma_t(u), \ t \in [0,u]\}$ is $\bG^N$-predictable, then $\ds \left\{\int_t^T  V(u,t^-)\Gamma_t(u)\ud u, \ t \in [0,T]\right\}$ is $\widetilde{\bH}$-predictable. Moreover,
\begin{align*}
&\espq{\int_0^T\left|\int_r^T  V(u,r)\Gamma_r(u)\ud u\right|\pi_{r}(\Lambda)\ud r}\\
&\quad \quad\leq \espq{\int_0^T\left\{\int_r^T  |V(u,r)||\Gamma_r(u)|\ud u\right\}\pi_{r}(\Lambda)\ud r}\\
&\quad \quad\leq \frac{1}{2} \espq{\int_0^T \left\{\int_r^T V^2(u,r)\ud u + \int_r^T \Gamma_r^2(u)\ud u\right\}\pi_{r}(\Lambda)\ud r}\\
&\quad \quad\leq \frac{1}{2} \espq{\int_0^T \left\{\int_0^u V^2(u,r)\pi_{r}(\Lambda)\ud r\right\}\ud u}+\frac{1}{2} \espq{\int_0^T  \left\{\int_0^u \Gamma_r^2(u)\pi_{r}(\Lambda)\ud r \right\}\ud u},
\end{align*}
where the last equality follows by Fubini's Theorem. Then we get that
\begin{align*}
&\espq{\int_0^T \left\{\int_0^u V^2(u,r)\pi_{r}(\Lambda)\ud r\right\}\ud u}\leq \espq{\int_0^T \left\{\sup_{0\leq r\leq u}V^2(u,r)\int_0^u \pi_{r}(\Lambda)\ud r\right\}\ud u}\\
&\quad \quad=\int_0^T\espu{\sup_{0\leq r\leq u}V^2(u,r)}\espd{\int_0^u \pi_{r}(\Lambda)\ud r}\ud u\\
&\quad \quad\leq l_a \int_0^T\espu{\sup_{0\leq r\leq u}V^2(u,r)}\ud u <\infty,
\end{align*}
thanks to \eqref{eq:integrab_intensita2}, $\ds \espd{\int_0^u \pi_{r}(\Lambda)\ud r}\leq l_a$, and by \eqref{eq:Vu} and the Burkholder-Davis-Gundy inequality, there exists a positive constant $c$ such that
\begin{align*}
\espu{\sup_{0\leq r\leq u}V^2(u,r)}\leq 4 \left(\espu{V_0^2}+c\espu{\int_0^u\beta^2_s(u)\ud \langle M \rangle_s}+ \left(\int_0^u|\beta_s(u)\alpha_s|\ud \langle M\rangle_s\right)^2+c\espu{\langle A^1 \rangle_u}\right)\\
\leq 4 \left(\espu{V_0^2}+c\espu{\int_0^T\beta^2_s(u)\ud \langle M \rangle_s}+ \left(\int_0^T|\beta_s(u)\alpha_s|\ud \langle M\rangle_s\right)^2+c\espu{\langle A^1 \rangle_T}\right)<\infty.
\end{align*}

Finally,
\begin{align*}
\espq{\int_0^T  \left\{\int_0^u \Gamma_r^2(u)\pi_{r}(\Lambda)\ud r \right\}\ud u}<\infty,
\end{align*}
since $\ds \espd{\int_0^u\Gamma^2_r(u)\pi_r(\Lambda)\ud r}<\infty$ for every $u \in [0,T]$ and the time interval $[0,T]$ is finite.
Indeed the term
\[
\left \{\int_0^t\left\{\int_r^T  V(u,r^-)\Gamma_r(u)\ud u\right\}\left(\ud N_r-\pi_{r^-}(\Lambda)\ud r\right),  \ t \in [0,T]\right\}
\]
is an $(\widetilde \bH, \Q)$-martingale.
These two martingales are orthogonal to $S$, because of the independence between $N$ and $S$.\\
It only remains to verify that the process $\ds \left\{\int_0^t\left\{\int_r^TB_{r^-}(u)\gamma_r(u)\ud u \right\}\ud A^1_r
, \ t \in [0,T]\right\}$  is an $(\widetilde \bH, \Q)$-martingale orthogonal to $S$. Since $\{B_{t^-}(u), \ t \in [0,u]\}$ is $\bG^N$-predictable and $\gamma(u)$ is $\bF$-predictable, then $\ds \left\{\int_t^T B_{t^-}(u)\gamma_t(u)  \ud u, \ t \in [0,T]\right\}$ is $\widetilde{\bH}$-predictable. Moreover,
\begin{align*}
&\espq{\int_0^T \left\{\int_r^TB_{r^-}(u)\gamma_r(u)\ud u\right\}^2 \ud \langle A^1\rangle_r}\leq
T\espq{\int_0^T \left\{\int_r^TB^2_{r^-}(u)\gamma^2_r(u)\ud u\right\} \ud \langle A^1\rangle_r}\\
&\quad =T\espq{\int_0^T \left\{\int_0^uB^2_{r^-}(u)\gamma^2_r(u)\ud \langle A^1\rangle_r\right\}\ud u }\\
&\quad \leq T\espq{\int_0^T \left\{\espd{\sup_{0\leq r\leq u}B^2_{r}(u)}\espu{\int_0^T\gamma^2_r(u)\ud \langle A^1\rangle_r}\right\}\ud u }<\infty
\end{align*}
since,  by \eqref{eq:disug_Bu}, we have $\ds \espd{\sup_{0\leq r\leq u}B^2_{r}(u)}<\infty$.

Therefore, $\ds \left\{\int_0^t\left\{\int_r^TB_{r^-}(u)\gamma_r(u)\ud u \right\}\ud A^1_r
, \ t \in [0,T]\right\}$ is an $(\widetilde \bH, \Q)$-martingale and
\[
\langle S, \int_0^\cdot \left\{\int_r^TB_{r^-}(u)\gamma_r(u)\ud u \right\}\ud A_r^1 \rangle_t =  \int_0^t  \left\{\int_r^TB_{r^-}(u)\gamma_r(u)\ud u \right\}\ud \langle S, A^1\rangle_r=0,
\]
due to the orthogonality between $S$ and $A^1$.

The square integrability of the martingale $K^1$ in \eqref{eq:mgK1} can be proved using similar computation to those used in the proof of Proposition \ref{prop:pseudo_opt_strategy}. In particular we get that $(K^1)^2$ is an $(\widetilde \bH, \Q)$-submartingale, and then

\[
\espq{(K_t^1)^2}\leq \espq{(K_T^1)^2}\leq3\left(\espq{G_T^2}+ \espq{\left(\int_0^T\left\{\int_r^TB_{r^-}(u)\beta_r(u) \ud u\right\}\ud S_r\right)^2}+\espq{G_0^2}\right)<\infty
\]
since $G_T\in L^2(\widetilde{\H}_T, \Q)$, $G_0\in L^2(\widetilde{\H}_0, \Q)$ and $\{\int_t^TB_{t^-}(u)\beta_t(u)\ud u, t \in [0,T]\}\in \Theta(\widetilde\bH)$.

\end{proof}

\end{document}